\documentclass[11pt]{article}

\usepackage{fullpage}
\usepackage{amsmath,amsthm,amsfonts,dsfont}
\usepackage{amssymb,latexsym,graphicx}
\usepackage{palatino}
\usepackage{mathpazo}
\usepackage{stmaryrd}
\usepackage{mathtools}
\usepackage{hyperref}
\usepackage[lined,boxed]{algorithm2e}
\usepackage{subfigure}
\usepackage{boxedminipage}

\newtheorem{theorem}{Theorem}[section]

\newtheorem{proposition}[theorem]{Proposition}
\newtheorem{lemma}[theorem]{Lemma}

\newtheorem{claim}[theorem]{Claim}

\newtheorem{corollary}[theorem]{Corollary}
\newtheorem{definition}[theorem]{Definition}

\newcommand{\N}{\ensuremath{\mathbb{N}}}
\newcommand{\F}{\ensuremath{\mathbb{F}}}
\newcommand{\R}{\ensuremath{\mathbb{R}}}
\newcommand{\Z}{\ensuremath{\mathbb{Z}}}

\newcommand{\Tr}{\mbox{\rm Tr}}

\newcommand{\cL}{\mathcal{L}}

 \newcommand{\eps}{\varepsilon} 
\renewcommand{\epsilon}{\varepsilon}

\renewcommand{\vec}[1]{\ensuremath{\mathbf{#1}}}

\newcommand{\basis}{\ensuremath{\mathbf{B}}}
\newcommand{\problem}[1]{\mbox{#1}\xspace}
\newcommand{\alg}[1]{\textup{\textsc{#1}}}
\newcommand{\poly}{\mathrm{poly}}
\newcommand{\polylog}{\mathrm{polylog}}

\DeclareMathOperator*{\expect}{\mathbb{E}}

\newcommand{\DGS}[3]{\ensuremath{#1\text{-}\problem{DGS}_{#2}^{#3}}}

\newcommand{\hDGS}[3]{\ensuremath{#1\text{-}\problem{hDGS}_{#2}^{#3}}}

\newcommand{\biglat}{\ensuremath{\mathcal{K}}}
\newcommand{\smalllat}{\ensuremath{\mathcal{L}}'}
\newcommand{\coset}{\ensuremath{\vec{c}}}

\newcommand{\bern}{\mathrm{B}}

\newcommand{\pois}{\mathrm{Pois}}
\newcommand{\Pois}{\pois}
\newcommand{\approxpmax}{\tilde{p}_{\mathsf{max}}}
\newcommand{\approxd}{\tilde{d}}
\newcommand{\pmax}{{p}_{\mathsf{max}}}
\newcommand{\pcol}{{p}_{\mathsf{col}}}
\newcommand{\psqrt}{{p}_{\mathsf{sqrt}}}
\newcommand{\pmin}{{p}_{\mathsf{min}}}

\newcommand{\scarequotes}[1]{``#1''}

\newcommand{\Tmax}{T_{\mathsf{max}}}
\newcommand{\Tmin}{T_{\mathsf{min}}}

\makeatletter
\def\imod#1{\allowbreak\mkern8mu({\operator@font mod}\,\,#1)}
\makeatother

\newcommand{\lat}{\mathcal{L}}
\newcommand{\gs}[1]{\ensuremath{\widetilde{#1}}}
\DeclareMathOperator{\dist}{dist}
\DeclareMathOperator{\spn}{span}
\DeclarePairedDelimiter\inner{\langle}{\rangle}
\DeclarePairedDelimiter\abs{\lvert}{\rvert}
\DeclarePairedDelimiter\set{\{}{\}}
\DeclarePairedDelimiter\floor{\lfloor}{\rfloor}
\DeclarePairedDelimiter\ceil{\lceil}{\rceil}
\DeclarePairedDelimiter\length{\lVert}{\rVert}

\begin{document}

\title{Solving the Shortest Vector Problem in $2^n$ Time via 
Discrete Gaussian Sampling}
\author{
Divesh Aggarwal\thanks{Department of Computer Science, EPFL.}
\and
Daniel Dadush\thanks{Centrum Wiskunde \& Informatica, Amsterdam.}\\
\texttt{dadush@cwi.nl}
\and
Oded Regev\thanks{Courant Institute of Mathematical Sciences, New York
 University.}
~\thanks{Supported by the Simons Collaboration on Algorithms and Geometry and by the National Science Foundation (NSF) under Grant No.~CCF-1320188. Any opinions, findings, and conclusions or recommendations expressed in this material are those of the authors and do not necessarily reflect the views of the NSF.}\\
\and
Noah Stephens-Davidowitz\footnotemark[3]
~\thanks{This material is based upon work supported by the National Science Foundation under Grant No.~CCF-1320188. Any opinions, findings, and conclusions or recommendations expressed in this material are those of the authors and do not necessarily reflect the views of the National Science Foundation.}\\
\texttt{noahsd@cs.nyu.edu}
}
\date{}
\maketitle

\begin{abstract}
We give a randomized $2^{n+o(n)}$-time and space algorithm for solving the
Shortest Vector Problem (SVP) on $n$-dimensional Euclidean lattices.  This
improves on the previous fastest algorithm: the deterministic $\widetilde{O}(4^n)$-time and $\widetilde{O}(2^n)$-space algorithm of Micciancio and Voulgaris (STOC
2010, SIAM J.\ Comp.\ 2013). 

In fact, we give a conceptually simple algorithm that solves the (in our opinion, even more interesting) problem of discrete Gaussian sampling (\problem{DGS}). More specifically, we show how to sample $2^{n/2}$ vectors from the discrete Gaussian distribution at \emph{any parameter} in $2^{n+o(n)}$ time and space. (Prior work only solved \problem{DGS} for very large parameters.) Our \problem{SVP} result then follows from a natural reduction from \problem{SVP} to \problem{DGS}. We also show that our \problem{DGS} algorithm implies a $2^{n + o(n)}$-time algorithm that approximates the Closest Vector Problem to within a factor of $1.97$.

In addition, we give a more refined algorithm for \problem{DGS}
above the so-called \emph{smoothing parameter} of the lattice, which 
can generate $2^{n/2}$ discrete Gaussian samples in just
$2^{n/2+o(n)}$ time and space. 
Among other things, this implies a $2^{n/2+o(n)}$-time and space algorithm
for $1.93$-approximate decision \problem{SVP}.

\end{abstract}


\thispagestyle{empty}

\newpage

\section{Introduction}
\label{sec:introduction}

A lattice $\lat$ is defined as the set of all integer combinations of some
linearly independent vectors $\vec{b}_1,\dots,\vec{b}_n \in \R^n$. The matrix
$\basis=(\vec{b}_1,\dots,\vec{b}_n)$ is called a basis of $\lat$, and we write
$\lat(\basis)$ for the lattice generated by $\basis$.

Perhaps the most central computational problem on lattices is the Shortest
Vector Problem ($\problem{SVP}$). Given a basis for a lattice $\lat \subseteq \R^n$,
$\problem{SVP}$ is to compute a non-zero vector in $\lat$ of minimum Euclidean norm. 

Starting in the '80s, the use of approximate and exact solvers for
$\problem{SVP}$ (and other lattice problems) gained prominence for their
applications in algorithmic number theory~\cite{LLL82}, coding over Gaussian
channels~\cite{Buda89},
cryptanalysis~\cite{Shamir84,Brickwell85,LagariasOdlyzko85}, combinatorial
optimization and integer programming~\cite{Lenstra83,Kannan87,FrankTardos87}.
Over the past decade and a half, the study of lattice problems greatly increased
due to newly found applications in cryptography. Many powerful cryptographic
primitives, such as fully homomorphic encryption~\cite{Gentry09,BV11,BrakerskiV14}, now
have their security based on the \emph{worst-case} hardness of approximating the
decision version of $\problem{SVP}$ (and other lattice problems) to within
polynomial factors~\cite{Ajtai96,MR04,oded05,BLPRS13}. 

From the computational complexity perspective, much is known about $\problem{SVP}$ in both its
exact and approximate versions. On the hardness side, \problem{SVP} was shown to be NP-hard
to approximate within any constant factor (under randomized reductions) and hard to approximate to
within $n^{c/\log \log n}$ for some constant $c > 0$ under reasonable
complexity assumptions~\cite{Mic01svp, Khot05svp, HRsvp}.
From the perspective of polynomial-time
algorithms, the celebrated LLL basis reduction gives a
$2^{O(n)}$ approximation algorithm for $\problem{SVP}$~\cite{LLL82}, and Schnorr's block reduction algorithm~\cite{Schnorr87}, with subsequent refinements~\cite{AKS01,MV13}, gives a
$r^{n/r}$ approximation in $2^{O(r)}\poly(n)$ time allowing for a smooth
tradeoff between time and approximation quality.  
  
As one would expect from the hardness results above, all known algorithms  for solving exact \problem{SVP}, 
including the ones we present here, require at least exponential time and 
sometimes also exponential space (and the same is true even for polynomial approximation factors). 
We mention in passing that despite running in exponential time,
these algorithms have practical importance in addition to the obvious theoretical importance. 
For instance, they are used for assessing the practical security 
of lattice-based cryptographic primitives, they are used as subroutines 
in the best current approximation algorithms (variants of block reduction),
and they are used in some applications where low-dimensional lattices 
naturally arise. 

While the state of the art for polynomial-time
approximation of lattice problems has remained relatively static over the last
two decades, the situation for exact algorithms has been
markedly different. Indeed, three major (and very different) classes of
algorithms for $\problem{SVP}$ have been developed. 

The first class, developed by Kannan~\cite{Kannan87} and refined by many
others~\cite{Helfrich86,HanrotStehle07,MicciancioWalter14}, is based on
combining strong basis reduction with exhaustive enumeration inside Euclidean
balls. The fastest current algorithm in this class solves $\problem{SVP}$ in
$\widetilde{O}(n^{n/(2e)})$ time while using $\poly(n)$ space~\cite{HanrotStehle07}. 

The next landmark algorithm, developed by Ajtai, Kumar, and
Sivakumar~\cite{AKS01} (henceforth AKS), is the most similar to this
work. AKS devised a method based on ``randomized sieving,'' whereby
exponentially many randomly generated lattice vectors are
iteratively combined to create shorter and shorter vectors, to give the first
$2^{O(n)}$-time (and space) randomized algorithm for $\problem{SVP}$. Many
extensions and improvements of their sieving technique have been proposed, both
provable~\cite{AKS02,MV10,PS09,LWXZ11} and
heuristic~\cite{NguyenVidick08,WLTB11,ZYH14,BGJ14,Laarhoven14}, where
the fastest provable sieving algorithm~\cite{PS09} for exact $\problem{SVP}$
requires $2^{2.465n+o(n)}$ time and $2^{1.233n+o(n)}$ space. 
It was observed by \cite{LWXZ11, priv:Micciancio, priv:Damien}
that AKS
can be modified to obtain a $2^{0.802 n + o(n)}$-time and $2^{0.401n +
o(n)}$-space algorithm for approximating $\problem{SVP}$ to within some large constant
factor.
Here $2^{.401n+o(n)}$ corresponds to
the best known upper bound on the $n$-dimensional ``kissing number'' (the
maximum number of points one can place on the unit sphere such that the pairwise
distances are $\geq 1$) due to Kabatjanski{\u\i} and Leven{\v{s}}te{\u\i}n~\cite{KL78}. 

The most recent breakthrough, due to Micciancio and Voulgaris~\cite{MV13}
(henceforth MV) and built upon the approach of Sommer, Feder, and
Shalvi~\cite{SFS09}, is a deterministic $\widetilde{O}(4^n)$-time and
$\widetilde{O}(2^n)$-space algorithm for $\problem{SVP}$. It uses the
\emph{Voronoi cell} of the lattice---the centrally symmetric polytope
corresponding to the points closer to the origin than to any other lattice point.

\paragraph{Main contribution.}  As our main result, we give a randomized
$2^{n+o(n)}$-time and space algorithm for exact SVP, improving on the
$\widetilde{O}(4^n)$ deterministic running time of MV.  
A second main result is a much faster $2^{n/2+o(n)}$-time (and space) algorithm that 
 approximates the decision version of $\problem{SVP}$ to within a small constant factor.

Our $2^{n+o(n)}$-time algorithm actually solves 
a more
difficult problem, namely, that of generating many discrete Gaussian samples from a
lattice with arbitrary parameter, as we describe below.  
We feel that this is even more interesting than the improved running time for
$\problem{SVP}$, and it should have further applications. As far as we are
aware, outside of security reductions having access to powerful oracles, this is
the first provable algorithm to use the discrete
Gaussian \emph{directly} to solve a classical lattice problem.  

\paragraph{Discrete Gaussian samplers.} 
Our first main technical
contribution is a general discrete Gaussian sampler, which will directly imply our $\problem{SVP}$ algorithm. Below, we give an informal description of this result. (See Section~\ref{sec:main-dgs} for the details.)

\begin{figure}[ht]
\begin{center}
\includegraphics[width=0.4 \textwidth]{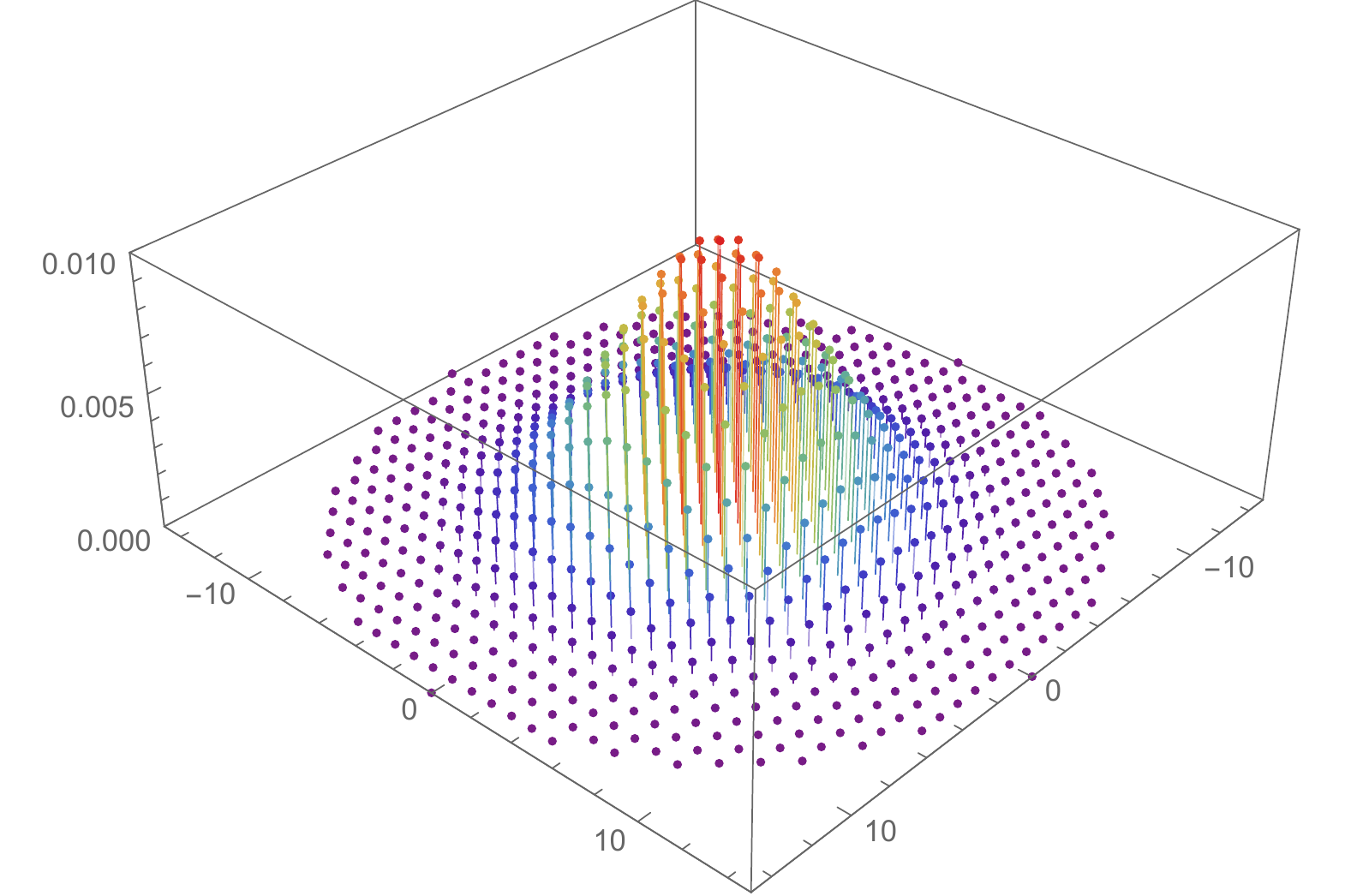}
\qquad
\includegraphics[width=0.4 \textwidth]{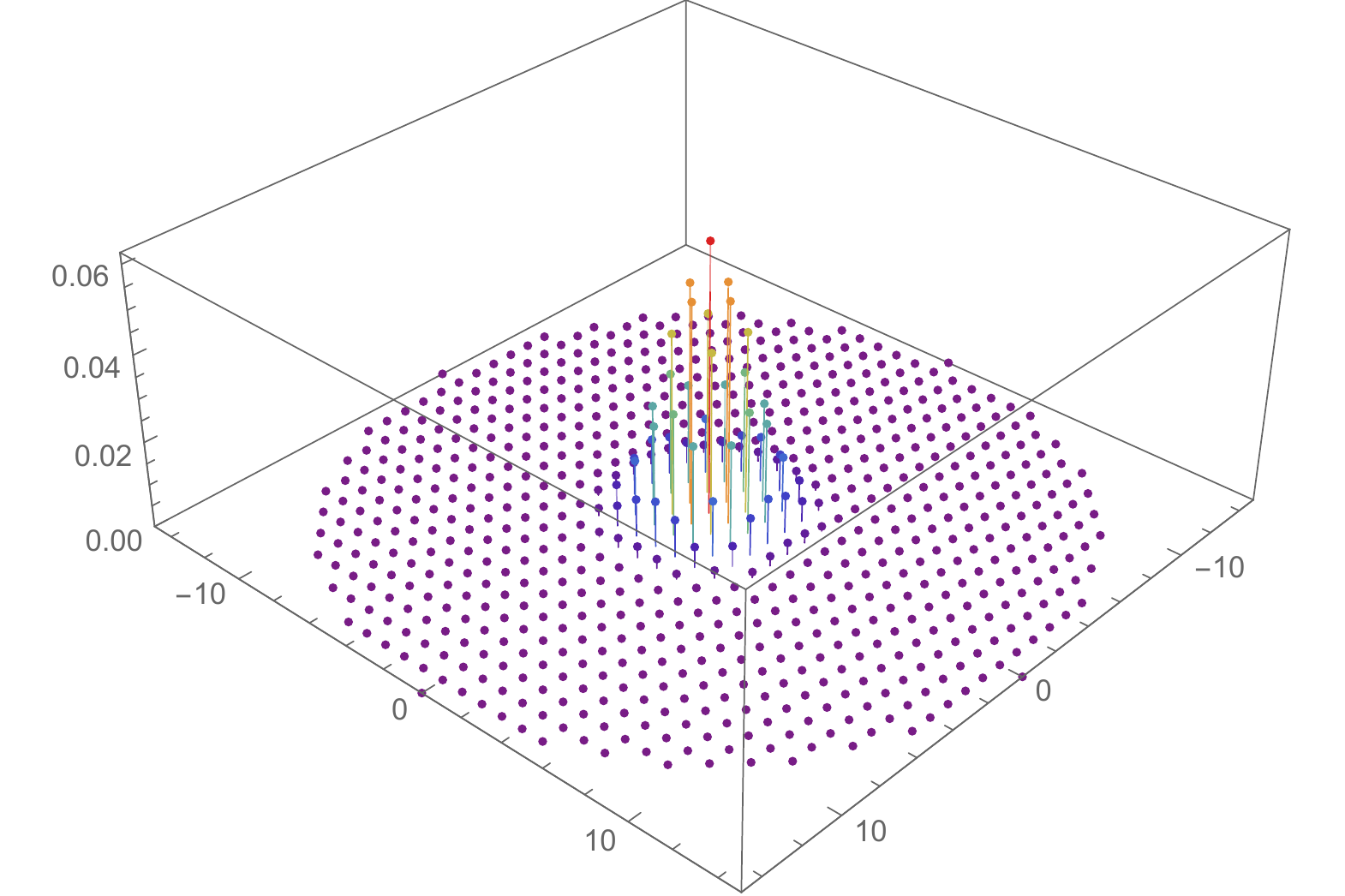}
\caption{\label{fig:DGS} 
The discrete Gaussian distribution on $\Z^2$ with parameter $s=10$ (left) and $s=4$ (right)}
\end{center}
\end{figure}

Define $\rho_s(\vec{x}) =
e^{-\pi\|\vec{x}\|_2^2/s^2}$ and $\rho_s(A) = \sum_{\vec{y} \in A} \rho_s(\vec{y})$
for any discrete set $A \subseteq \R^n$. The discrete Gaussian distribution
$D_{\lat,s}$ over the lattice $\lat \subseteq \R^n$ with parameter $s$
is the distribution satisfying 
\[
\Pr_{\vec{X} \sim D_{\lat,s}}[\vec{X}=\vec{x}] = \rho_s(\vec{x})/\rho_s(\lat),\quad 
\forall \vec{x} \in \lat \text{.}
\]
See Figure~\ref{fig:DGS} for an illustration.
The parameter $s$ determines the ``width'' of the discrete Gaussian. Note that
as $s$ becomes smaller, $D_{\lat,s}$ becomes more and more concentrated on short
lattice vectors. Hence it should not come as a surprise that being able to obtain sufficiently many samples 
from $D_{\lat,s}$ for an arbitrary $s$ leads to a solution to \problem{SVP}. We will discuss 
this relatively natural reduction below, but first let us describe our main
technical contributions, the Gaussian samplers. 

\begin{theorem}[General discrete Gaussian Sampler, informal] 
There is an algorithm that takes as input a lattice $\lat \subset \R^n$ and \emph{any parameter} $s >0$ and outputs 
$2^{n/2}$ i.i.d.\ samples from $D_{\lat,s}$ using $2^{n+o(n)}$ time and space.
\label{thm:gen-dgs-inf}
\end{theorem}

Notice the amortized aspect of the algorithm: we obtain $2^{n/2}$ vectors in about $2^n$ time. We do not know how to reduce the time to 
$2^{(1- \eps)n}$
---even if all we want is just one vector!
(But see below for a faster algorithm that works for large parameters.)
Improving the running time of the algorithm (while still outputting a sufficiently large number of samples) would immediately translate into an improved \problem{SVP} algorithm.

As we explain below, a closer inspection of the technique used in our algorithm 
suggests that with some refinement it might be able to achieve a running time of $2^{n/2}$.
Indeed, we actually do achieve this, but only for sufficiently large parameters $s$. In fact, for such parameters, we actually manage to output $2^{n/2}$ samples in $2^{n/2+o(n)}$ time.
This is our second main technical contribution.

\begin{theorem}[Smooth discrete Gaussian sampler, informal] 
There is an algorithm that takes as input a lattice $\lat \subset \R^n$ and a parameter $s $ \emph{above the smoothing parameter} of $\lat$ and outputs 
$2^{n/2}$ i.i.d.\ samples from $D_{\lat,s}$ using $2^{n/2+o(n)}$ time and space.
\label{thm:smooth-dgs-inf}
\end{theorem}

The smoothing parameter is the value of $s$ above which $D_{\lat,s}$ ``looks like'' a continuous Gaussian 
in a certain precise mathematical sense. (See
Definition~\ref{def:smooth}.) While sampling above smoothing appears to be insufficient 
to solve exact lattice problems, it is enough to solve major lattice
problems approximately. 
Indeed, we show how this is sufficient to approximate the decision version of \problem{SVP} 
to within a constant factor in time $2^{n/2+o(n)}$ (with the constant being roughly $1.93$). 
This holds the record for the fastest provable running time of a hard
lattice problem. 

We note that Theorem~\ref{thm:smooth-dgs-inf} actually implies a slightly stronger result; we can obtain $2^{n/2}$ samples from the Gaussian over a lattice \emph{shift}, $D_{\lat - \vec{t}, s}$, in $2^{n/2 + o(n)}$ time and space, as long as $s$ is above smoothing. However, we know of no application of this slightly stronger result. (See Section~\ref{sec:shifted} for a proof sketch.)

\subsection{Comparison with prior work}

The task of discrete Gaussian sampling is by no means new.
It by now has a long history
within cryptography \cite{MR04,GPV08,oded05,Peikert09, MicciancioP13}.
By analyzing an algorithm of Klein~\cite{Klein00}, 
Gentry, Peikert, and Vaikuntanathan~\cite{GPV08} first showed how to solve \problem{DGS} in polynomial time for large parameters.
(We remark that Klein analyzed this algorithm for very small parameters
and used it to solve the BDD problem. For such parameters, the algorithm does not produce samples distributed according to the discrete Gaussian distribution.)
\problem{DGS}
has been used extensively to improve reductions from worst-case lattice problems (such as
approximate decisional $\problem{SVP}$) to the average-case Short Integer Solution (SIS) and
Learning with Errors (LWE) problems \cite{MR04,oded05,Peikert09, MicciancioP13}, and as a core subroutine for
instantiating certain cryptographic primitives~\cite{GPV08}. In all
previous works, the \problem{DGS} procedure either samples at very high parameters or requires \emph{a priori knowledge of a
relatively short lattice basis}---typically only available when a user is able to generate
the lattice themselves, such as in certain trapdoor schemes---or access to
\emph{powerful oracles}, such as SIS or LWE oracles.

Furthermore, even with oracles and a short basis, none of the algorithms from prior work
could be used to sample below the \emph{smoothing
parameter} of the lattice.
The reason that we are able to achieve this is because of our observation that, if we allow ourselves exponential time, we can carefully combine vectors sampled from a discrete Gaussian together to obtain vectors whose distribution is \emph{exactly} a discrete Gaussian with a smaller parameter. (See Lemma~\ref{lem:sumofgaussians} and the proof overview below.) All prior work only obtained a distribution that is \emph{statistically close} to the discrete Gaussian, with error that is unbounded below the smoothing parameter.

We note that our approach is similar to that of the AKS algorithm at a high level. In particular, like AKS, we use a sieve algorithm that starts with a large collection of randomly selected vectors and proceeds to combine them together in pairs to find short lattice vectors. The major important difference between our approach and that of the AKS algorithm and its derivatives is that we maintain complete control over the distribution of the lattice
points that we generate at each step. While prior work is focused (quite naturally) on controlling the \emph{lengths} of the vectors after each step, our algorithm actually completely ignores their lengths---choosing whether to combine two vectors based only on their \emph{coset} mod a sublattice. 

Indeed, we view our $2^{n+o(n)}$-time algorithm as an efficient discrete Gaussian sampler that
consequently yields an efficient solution to \problem{SVP}, rather than as a
sieve algorithm for \problem{SVP}. It is the simplicity and elegance of the
discrete Gaussian distribution that allows us to side-step many of the
complications that arise with other sieve algorithms (such as the
\scarequotes{perturbation} step). Indeed, the $2^{n+o(n)}$-time
algorithm is quite simple; the most technical tool that it uses is 
a simple subroutine that we call
the \scarequotes{square sampler} (described below).

One negative aspect of our approach is that it has a clear lower bound. It seems that we cannot use this approach to find any algorithm that runs in time less than $2^{n/2}$. And, the quoted running time of each algorithm ($2^{n+o(n)}$ and $2^{n/2 + o(n)}$ respectively) is essentially tight in both theory and practice---for large (and relevant) parameters, our 
sieves yield essentially nothing when their input consists of fewer than $2^n$ or $2^{n/2}$ vectors respectively. This is in contrast to AKS-style algorithms, which seem to perform well heuristically~\cite{NguyenVidick08,WLTB11,ZYH14,Laarhoven14}. 

\subsection{Proof overview}

We now include a high-level description of our proofs, first that of Theorem~\ref{thm:gen-dgs-inf} and then that of the more refined Theorem~\ref{thm:smooth-dgs-inf}. We end with a brief discussion of how to use Gaussian samples
to solve \problem{SVP} as well as other applications.

\paragraph{A $2^{n+o(n)}$-time combiner for DGS.}
Recall that efficient algorithms are known for sampling from the discrete Gaussian at very high parameters \cite{GPV08}.
It therefore suffices to find a way to efficiently \emph{convert} samples from the discrete Gaussian with a high parameter to samples with a parameter lowered by a constant factor. By repeating this \scarequotes{conversion} many times, we can obtain samples with much lower parameters.

Note that this is trivial to do for the continuous Gaussian: if we divide a vector sampled from the continuous Gaussian distribution by $2$, the result is distributed as a continuous Gaussian with half the width. Of course, half of a lattice vector is not typically in the lattice, so this method fails spectacularly when applied to the discrete Gaussian. But, we can try to fix this by \emph{conditioning} on the result staying in the lattice. I.e., we can sample many vectors from $D_{\lat, s}$, keep those that are in the \scarequotes{doubled lattice} $2\lat$, and divide them by two. This method does work, but it is terribly inefficient---there are $2^n$ cosets of $2\lat$, and for some typical parameters, a sample from $D_{\lat, s}$ will land in $2\lat$ with probability as small as $2^{-n}$. I.e., our \scarequotes{loss factor,} the ratio of the number of output vectors to the number of input vectors, can be as bad as $2^{-n}$ for a single step. If we wish to iterate this $k$ times, we could need $2^{kn}$ input vectors for each output vector, resulting in a very slow algorithm!

We can be much more efficient, however, if we instead look for \emph{pairs} of vectors sampled from $D_{\lat, s}$ whose \emph{sum} is in $2\lat$, or equivalently pairs of vectors that lie in the same coset $\vec{c}$ mod $2\lat$. Taking our intuition from the continuous Gaussian, we might hope that the \emph{average} of two such vectors will be distributed as $D_{\lat, s/\sqrt{2}}$. This suggests an \emph{amortized} algorithm, in which we sample many vectors from $D_{\lat,s}$, place them in \scarequotes{buckets} according to their coset mod $2\lat$, and then take the average of disjoint pairs of elements in the same bucket. We call such an algorithm a \scarequotes{combiner.} The most natural combiner to consider is the \scarequotes{greedy combiner,} which simply pairs as many vectors in each bucket as it can, leaving at most one unpaired vector per bucket. Since there are $2^n$ cosets, if we take, say, $\Omega(2^n)$ samples from $D_{\lat, s}$, almost all of the resulting vectors will be paired. A lemma due to Peikert (\cite{Pei10}) shows that the resulting distribution will be statistically close to the desired distribution, $D_{\lat, s/\sqrt{2}}$, \emph{provided that the parameter $s$ is above the smoothing parameter}.

At this point, we can already build a roughly $2^n$-time algorithm for \problem{DGS} that works for such parameters. (Namely, use prior work to sample at some very high parameter and iteratively apply the combiner described above.) While this is not our main result (it is strictly weaker), we note that we have not seen this observation mentioned elsewhere.\footnote{One can likely also obtain a $2^{O(n)}$-time algorithm for $\problem{DGS}$ above the smoothing parameter by instantiating the oracles in~\cite{MicciancioP13}.} But, in order to move \emph{below smoothing} (which is necessary, e.g., for solving \problem{SVP}), we need to do something else.

\begin{figure}[ht]
\begin{center}
\[
\hspace{-3ex}
\raisebox{14 ex}{\includegraphics[width=0.1 \textwidth]{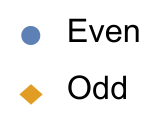}}
\includegraphics[width=.40 \textwidth]{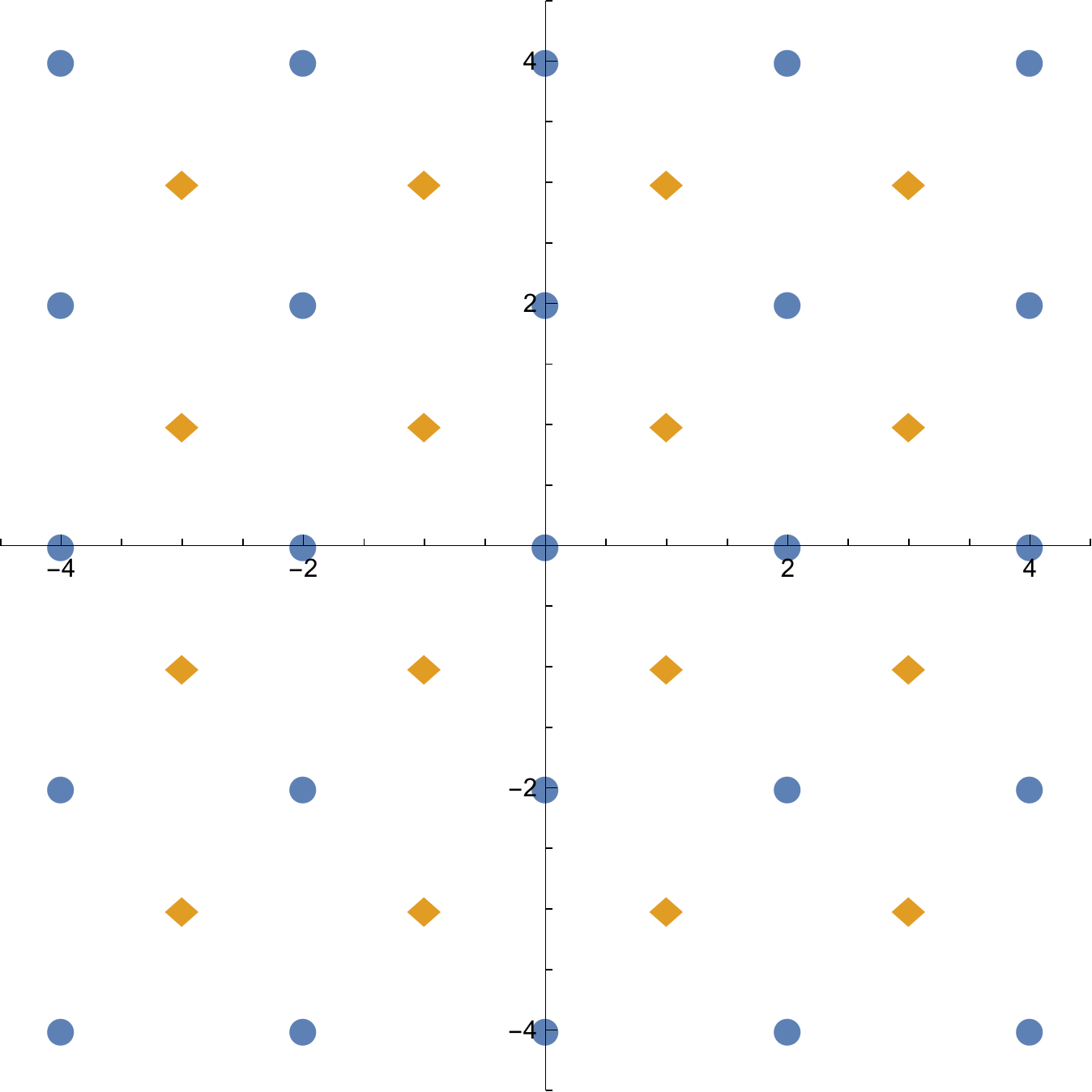}
\quad
\raisebox{10 ex}{\includegraphics[width = .1 \textwidth]{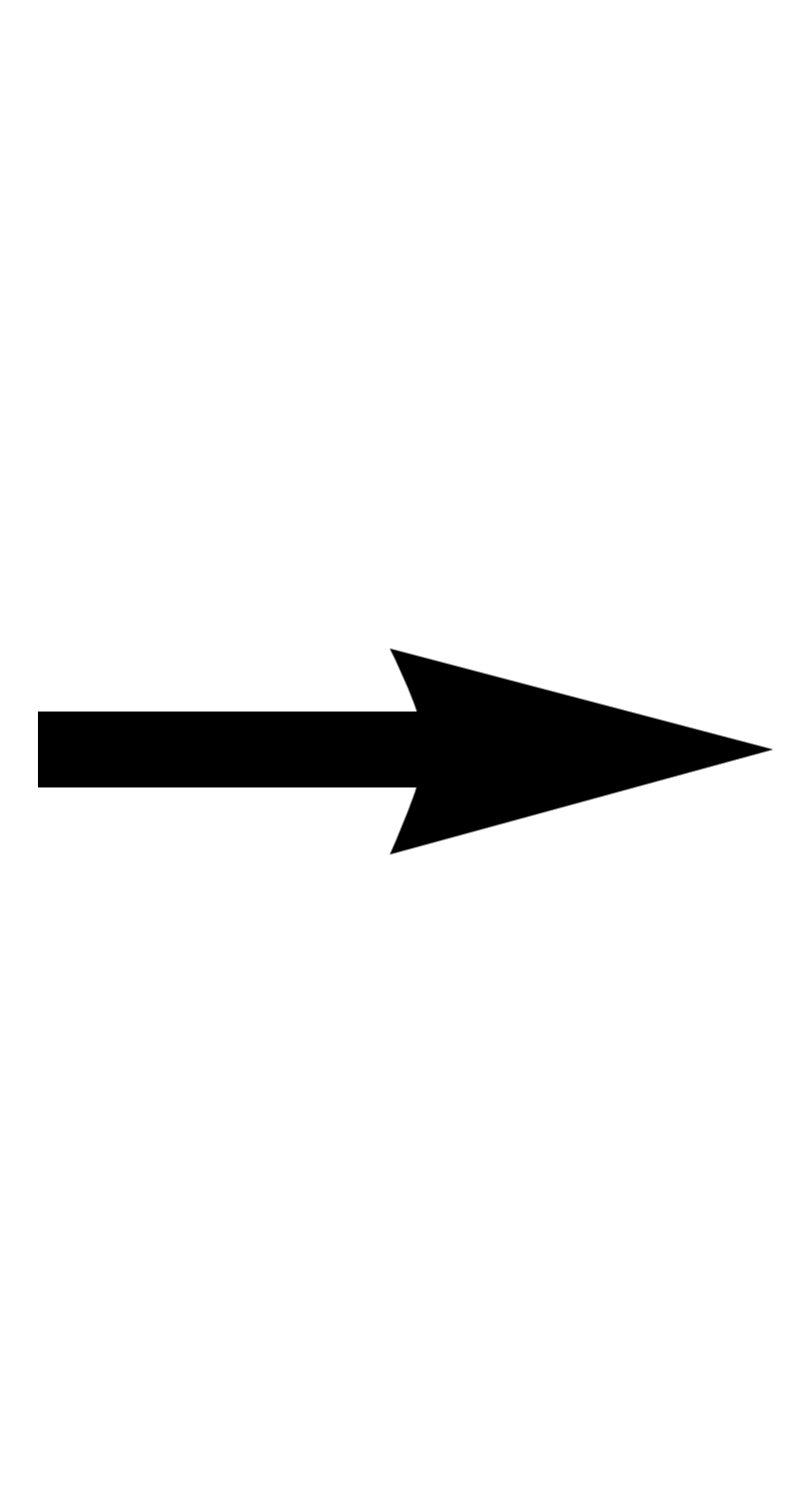}} %
\quad
\includegraphics[width=0.40 \textwidth]{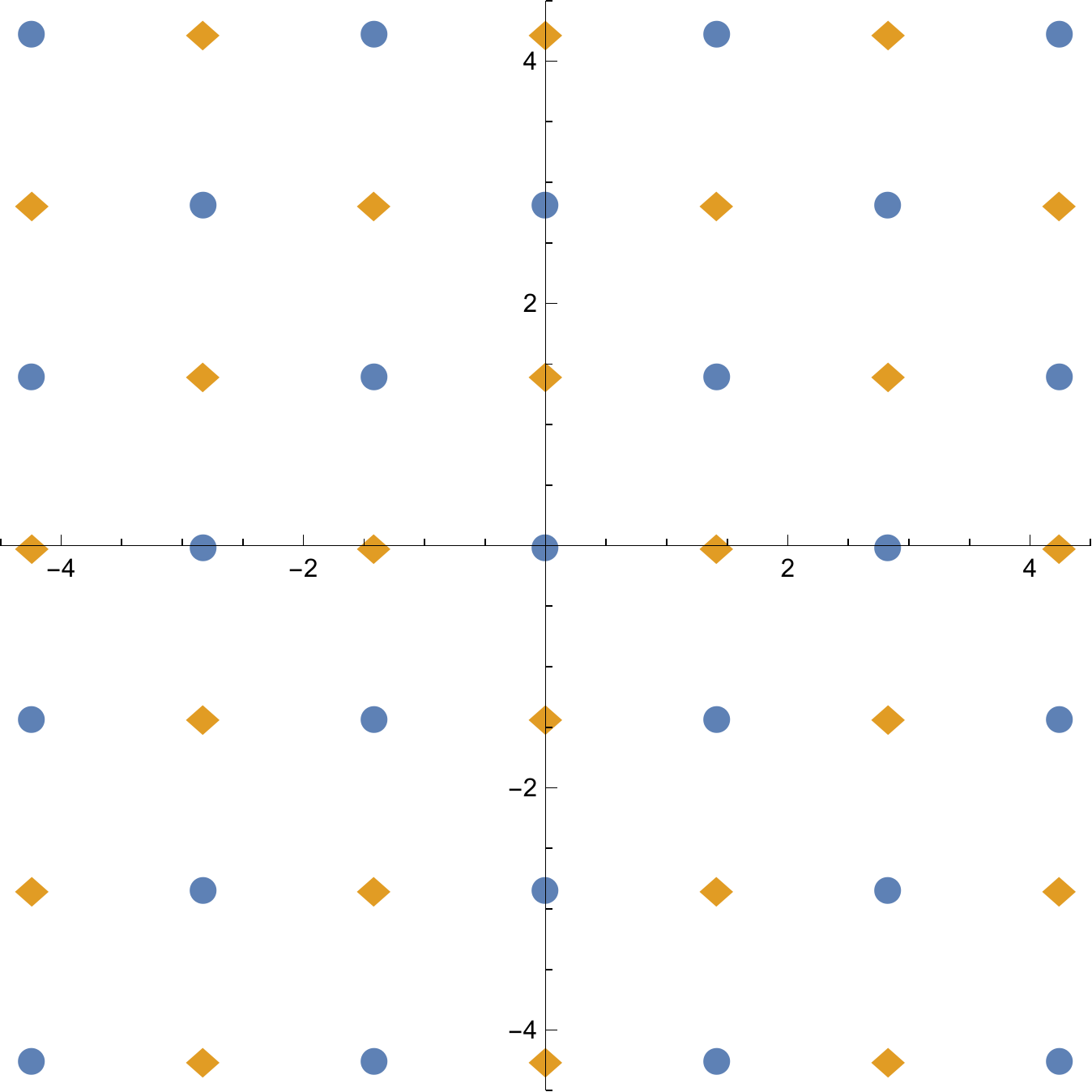}
\]
\caption{\label{fig:rotation} The rotation mapping $\overline{\lat} = \{ (z_1, z_2) \in \Z^2 : z_1 = z_2 \bmod 2\}$ to $(\sqrt{2} \Z)^2$, $(z_1, z_2) \mapsto (z_1 + z_2, z_1-z_2)/\sqrt{2}$. Our algorithm effectively creates samples from $D_{\overline{\lat}, s}$ and then outputs the first coordinate of the rotated result scaled down by a factor of $\sqrt{2}$, $(z_1 + z_2)/2$. The resulting distribution is exactly $D_{\Z,s/\sqrt{2}}$.
}
\end{center}
\end{figure}

\begin{figure}[ht]
\begin{center}
\includegraphics[width=0.8 \textwidth]{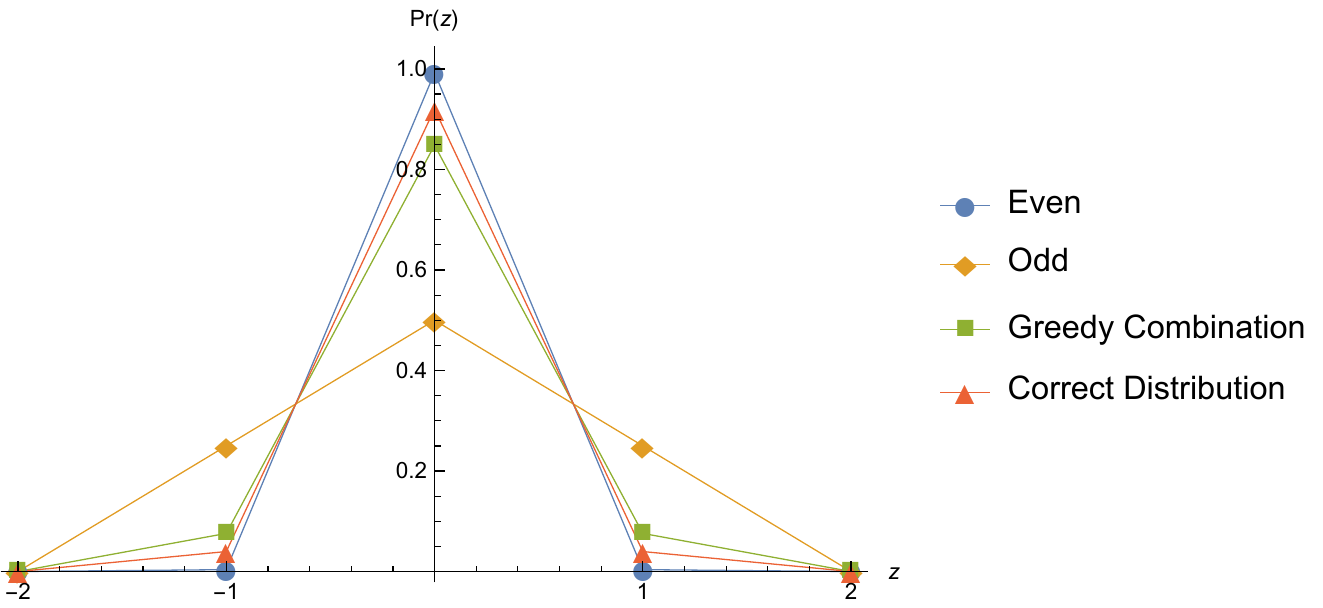}
\caption{\label{fig:greedybargraph}
The distribution of averages of pairs of integers sampled from $D_{\Z, \sqrt{2}}$ resulting from taking (1) only even pairs; (2) only odd pairs; (3) even and odd pairs with \scarequotes{greedy} weights proportional to $\rho_{\sqrt{2}}(2\Z)$ and $\rho_{\sqrt{2}}(2\Z+1)$ respectively; and (4) even and odd pairs with \scarequotes{squared} weights proportional to $\rho_{\sqrt{2}}(2\Z)^2$ and $\rho_{\sqrt{2}}(2\Z+1)^2$ respectively. The fourth distribution is exactly $D_{\Z}$.}
\end{center}
\end{figure}

In particular, below the smoothing parameter, combining discrete Gaussian vectors \scarequotes{greedily} as above will not typically give a result that is statistically close to a Gaussian distribution. However, all is not lost. 
Recall that our algorithm works by picking pairs of vectors sampled independently from $D_{\lat, s}$ that are in the same coset $\vec{c}$ mod $2\lat$, and then taking the average of each pair. 
So, the algorithm effectively samples a vector $(\vec{X}_1, \vec{X}_2)$ from \emph{some} distribution over the $2n$-dimensional lattice of pairs of vectors from $\lat$ that are in the same coset mod $2\lat$, 
\[
\overline{\lat} := \{ (\vec{X}_1, \vec{X}_2) \in \lat^2\ :\ \vec{X}_1 = \vec{X}_2 \bmod 2\lat \} = \bigcup_{\vec{c} \in \lat/(2\lat)} \vec{c} \times \vec{c}
\; ,
\] 
and then outputs $(\vec{X}_1+\vec{X}_2)/2$.
We claim that \emph{assuming that that distribution is} $D_{\overline{\lat},s}$, 
the output $(\vec{X}_1+\vec{X}_2)/2$ is distributed \emph{exactly}
as $D_{\lat, s/\sqrt{2}}$. 
This fact, shown in Lemma~\ref{lem:sumofgaussians}, has a straightforward proof, yet we have not seen this observation before. 
(It is closely related to Riemann's theta relations, as described 
in~\cite[Chapter 1, Section 5]{Mumford}; see also~\cite{RegevS15})
The idea is the following.
It is not difficult to show that if we apply the $45^\circ$ rotation given by 
\[ 
(\vec{X}_1, \vec{X}_2) \mapsto \Big(\frac{\vec{X}_1 + \vec{X}_2}{\sqrt{2}}, \frac{\vec{X}_1 - \vec{X}_2}{\sqrt{2}} \Big)
\] 
to $\overline{\lat}$,
we obtain the \emph{product} lattice $(\sqrt{2} \lat)^2 = \{ (\sqrt{2} \vec{X}_1, \sqrt{2}\vec{X}_2)~|~\vec{X}_1,\vec{X}_2 \in \lat\}$. 
(Figure~\ref{fig:rotation} shows the one-dimensional case.)
Note that the rotation of a discrete Gaussian is again a discrete Gaussian, and the discrete Gaussian over a product lattice is a product distribution.
Therefore, if $(\vec{X}_1, \vec{X}_2)$ is distributed according to $D_{\overline{\lat}, s}$, 
the distribution of $(\vec{X}_1 + \vec{X}_2,\vec{X}_1 - \vec{X}_2)/\sqrt{2}$ is
according to $D_{(\sqrt{2} \lat)^2,s}$ which is a product distribution,
and hence $(\vec{X}_1 + \vec{X}_2)/\sqrt{2}$ is distributed according
to $D_{\sqrt{2} \lat,s}$, as claimed. 

However, note that if the combiner just 
greedily paired as many vectors from each coset as possible, 
it would \emph{not} yield samples from $D_{\overline{\lat},s}$. In 
particular, the probability that a sample from $D_{\overline{\lat}, s}$ 
will land in $\vec{c} \times \vec{c}$ for some coset $\vec{c}$ 
is proportional to the \scarequotes{squared weight} of the 
coset $\rho_s(\vec{c})^2$. But, the greedy approach pairs vectors
from $\coset$ with probability roughly proportional to 
$\rho_s(\coset)$. (Figure~\ref{fig:greedybargraph} shows 
how the resulting distributions differ in the one-dimensional 
case.) For parameters above smoothing, these distributions 
are roughly the same, but to go below smoothing (and to 
avoid the statistical error resulting from the greedy 
approach), we need a way to sample pairs from this 
\scarequotes{squared distribution} directly.

This mismatch between the \scarequotes{squared distribution} that we want and the \scarequotes{unsquared} distribution that we get is the primary technical challenge that we must overcome to build our general discrete Gaussian combiner. To solve it, we present a generic solution for
\scarequotes{converting any probability distribution to its square} relatively
efficiently, which we call the \scarequotes{square sampler.} Informally, the square sampler is given access to samples from some probability distribution that assigns respective (unknown) probabilities $(p_1,\ldots, p_N)$ to the elements in some (large) finite set $\{1,\ldots, N\}$. It uses this to efficiently sample a large collection of \emph{independent} coin flips $b_{i,j}$ such that $b_{i,j} = 1$ with probability proportional to $p_i$.  Then, using these coins, it applies rejection sampling to the input samples (accepting the $j$th instance of input value $i$ if $b_{i,j} = 1$) in order to obtain the desired \scarequotes{squared distribution.} If $\Pr[b_{i,j} = 1] = T p_i$ for some proportionality factor $T$, it is not hard to see that the expected \scarequotes{loss factor} 
of this process is $T\sum p_i^2$. We therefore take $T$ to be as large as possible by setting $T \approx 1/\max p_i$ (if we took $T$ to be any larger, we would need a coin that lands on heads with probability greater than one!), making the loss factor of the square sampler approximately $\sum p_i^2/\max p_i$. (See Section~\ref{sec:squaresampler} and Theorem~\ref{thm:squaresampler} in particular.)

In particular, when combining discrete Gaussian vectors, the loss factor is approximately the \emph{collision probability} over the cosets $\vec{c}$ of
$2\lat$, $\sum \rho_s(\vec{c})^2/\rho_s(\lat)^2$, divided by the maximal
probability of a single coset. As a result, if one
coset has a  $2^{-n/2}$ fraction of the total weight and the other cosets split the remaining weight
roughly evenly, then the loss factor is roughly $ 2^{-n/2}$ \emph{for a single
step of the combiner}. This looks terrible for us, as it could be the case that $k$ applications of the combiner could yield a loss factor of $2^{-kn/2}$! Surprisingly, we
show that the product of all loss factors for an arbitrarily long sequence of
applications of the combiner is at worst $2^{-n/2}$ (ignoring loss due to other
factors). I.e., the accumulated loss factor can be no worse than essentially the
worst-case loss factor in a single step!\footnote{While the purely algebraic proof of
this fact is quite simple (see the proof of Corollary~\ref{cor:pipeline}), we do not yet
have good intuitive understanding of it. Indeed, we have found ourselves referring to it as the \scarequotes{magic cancellation.}} As a result, our general combiner
always returns $2^{n/2}$ vectors when its input is $2^{n + o(n)}$ vectors
sampled from the discrete Gaussian. (See Corollary~\ref{cor:pipeline} for the formal analysis of repeated application of our combiner.)

\paragraph{A $2^{n/2+o(n)}$-time combiner for DGS above smoothing.} Recall that the general combiner described above starts with many vectors and then repeatedly takes the average of pairs of vectors that lie in the same coset of $2\lat$. We observed that this combiner necessarily needs over $ 2^n$ vectors \scarequotes{just to get started} because it works over the $2^n$ cosets of $2\lat$. To get a faster combiner, we therefore try pairing vectors according to the cosets of some sublattice $\lat'$ that \scarequotes{lies between} $\lat$ and $2\lat$ such that $2\lat \subseteq \lat' \subset \lat$. If we simply take many samples from $D_{\lat, s}$, group them according to their cosets mod $\lat'$, and sum them together (taking averages is a bit less natural in this context), analogy with the continuous Gaussian suggests that the resulting vectors will be distributed as roughly $D_{\lat', \sqrt{2}s}$. 
Note that the parameter has increased, which is not what we wanted, but we are now sampling from a sparser lattice. In particular, suppose that we apply this combiner twice, so that in the second step we obtain vectors from some sublattice $\lat''$. We then expect to obtain samples from roughly $D_{\lat'', 2s}$. So, intuitively, if we take $\lat''$ to be a sublattice of $2\lat$, we have \scarequotes{made progress,} even though we have doubled the parameter. Our running time will be proportional to the index of $\lat'$ over $\lat$ (assuming that the index of $\lat''$ over $\lat'$ is the same), so we should take the index of $\lat'$ over $\lat$ to be as small as possible. 
More specifically, we can build a \scarequotes{tower} of progressively sparser lattices $(\lat_0,\ldots, \lat_\ell)$ with the index of $\lat_i$ over $\lat_{i-1}$ taken to be slightly larger than $2^{n/2}$.\footnote{We note that Becker et al.~\cite{BGJ14} also use a tower of lattices in their heuristic algorithm.
}
If we take $\lat_\ell$ to be the lattice from which we wish to obtain samples with parameter $s$ and $\lat_0$ to be a dense lattice from which we can sample efficiently with parameter $2^{-\ell/2}s$, we can hope that iteratively applying such a combiner \scarequotes{up the tower} will yield a sampling algorithm.

As in the description of our $2^n$-time combiner, 
the lemma from~\cite{Pei10} shows that the above approach, when instantiated with the \scarequotes{greedy combiner,} will yield an algorithm that can output vectors whose distribution is statistically close to the discrete Gaussian for parameters $s$ that are above the smoothing parameter. 
Though this statistical distance can be made small, 
it is large enough to break applications such as our approximation 
algorithm for decision $\problem{SVP}$.

To avoid this error, the natural hope is that the same combiner used in the $2^n$-time algorithm above (the one with the \scarequotes{square sampler}) will suffice.
Unfortunately, this gives the wrong distribution. In particular, we obtain a distribution in which the cosets of $\lat'$ over $2\lat$ have weight that is proportional to the \emph{square} of their weights over the discrete Gaussian. (See Lemma~\ref{lem:anysublattice}. 
Note that when $\lat' = 2\lat$ there is only one such coset, which is why our $2^n$-time combiner does not run into this problem.) In some sense, this is the \scarequotes{inverse} of the problem that the square sampler solves. And, indeed, we solve it by building a \scarequotes{square-root sampler}---based on a clever trick (used implicitly in~\cite{MosselP05} and discussed in~\cite{DidMSE}) that allows one to flip a coin with probability $\sqrt{p}$ given black-box access to a coin with unknown probability $p$. (See Claim~\ref{clm:sqrt} for the trick and Theorem~\ref{thm:sqrtsampler} for the square-root sampler.) So, our combiner works by \scarequotes{squaring the weights} of the cosets mod $\lat'$ of input vectors; pairing them according to these squared weights and summing the pairs; and then \scarequotes{taking the square root of the weights} of the cosets mod $2\lat$ of the resulting output vectors.

This completes the description of the proof of Theorem~\ref{thm:smooth-dgs-inf}. The above only works above the smoothing parameter  because the required size of the input to the square-root sampler depends on $1/\pmin$, where $\pmin$ is the probability of landing in the coset with minimal weight. We therefore only know how to use the square-root sampler to efficiently sample above the smoothing parameter, where the minimal weight is roughly equal to the maximal weight. Indeed, in this regime, both the square-root sampler and the square sampler incur \scarequotes{almost no loss,} so that we obtain an algorithm that runs in time $2^{n/2 + o(n)}$ and returns $2^{n/2}$ samples from the discrete Gaussian. But, below this, $1/\pmin$ can be arbitrarily large. (Intuitively, some sort of dependence on $1/\pmin$ is necessary for a square-root sampler because a coset whose weight is negligible could have significant weight after we \scarequotes{take the square root.} So, we should expect that any square-root sampler would need at least enough samples to \scarequotes{see} such a coset.)

However, we again stress that our techniques do not incur error that depends on \scarequotes{how smooth the distribution is.} 
This leaves open the possibility that our algorithm might be modified to work even \emph{below} smoothing. The only bottleneck is that the square-root sampler requires very large input in such cases. But, we note that the way that we currently use the square-root sampler might not be optimal.
More specifically, we observe the rather strange behavior of 
our current algorithm: when the algorithm \scarequotes{takes the square root} of some coset weights, it typically \scarequotes{squares} the weights of some (different!) cosets immediately afterwards. 
 So, while the second step is not the exact inverse of the first, it does still seem that the square-root step is a bit counterproductive. 
This suggests that there is room for improvement in this algorithm,
and we have made some progress to that end
by proving a correlation inequality that we believe should play
a central role in such an improved algorithm~\cite{RegevS15}. 

Finally, we note that this algorithm can actually be used to obtain $2^{n/2}$ samples from the \emph{shifted} discrete Gaussian $D_{\lat - \vec{t}, s}$ for any $\vec{t} \in \R^n$ and parameter $s$ above smoothing in $2^{n/2+o(n)}$ time. We know of no applications for this, but we present a proof sketch in Section~\ref{sec:shifted} for completeness.

\paragraph{Reduction from SVP to DGS.} 
As we mentioned above, if we could efficiently
sample from $D_{\lat,s}$ at the \emph{right parameter}, we can hope that a
$D_{\lat,s}$ sample will hit a shortest non-zero vector of $\lat$ with
reasonable probability. 
One can quickly see here that there is an important trade-off in the choice of
$s$. For $s$ too small, the discrete Gaussian becomes completely concentrated on
$\vec{0}$, whereas for $s$ too large, the distribution becomes too diffuse
over $\lat$ and will rarely hit a shortest vector. 
By properly choosing $s$, we show that the discrete Gaussian yields a shortest non-zero lattice vector with probability $2^{-0.465n-o(n)}$. (In Section~\ref{sec:boundonmass}, we note that the optimal parameter has a nice interpretation in terms of the smoothing parameter.) Since our \problem{DGS} algorithm returns $2^{n/2}$ vectors in time $2^{n + o(n)}$, we obtain a $2^{n + o(n)}$-time algorithm for \problem{SVP}.

In order to obtain this bound, we use the result of Kabatjanski{\u\i} and Leven{\v{s}}te{\u\i}n that achieves the current best upper bound on the kissing number~\cite{KL78}. (The kissing number bounds from above the maximal number of shortest non-zero vectors in a lattice. Note that the reciprocal of the latter is a natural upper bound on the above probability.) At a high level, this is essentially the same problem faced by the randomized
sieving algorithms, and our techniques are very similar to those developed there (in
particular those in~\cite{PS09,MV10}).

\paragraph{Reduction from decision SVP to DGS above smoothing.} 
In order to approximate the length of the shortest non-zero lattice vector to within a constant factor, we note (in Lemma~\ref{lem:etalambda1}) that it suffices to approximate the smoothing parameter of the dual lattice (for exponentially small $\eps$) to within a constant factor. Of course, if we had a $2^{n/2}$-time discrete Gaussian sampler that worked above smoothing and always failed below smoothing, then it would be trivial to use this to approximate the smoothing parameter. However, while our $2^{n/2 + o(n)}$-time sampler does in fact always work above smoothing, 
it is not a priori clear how it behaves when asked to provide samples below smoothing. 

We handle this problem as follows. First, while we cannot guarantee that our sampler always fails below smoothing, we show (with a bit more work) that it always either fails or outputs valid discrete Gaussian samples. We call such a sampler \scarequotes{honest.} (See Definition~\ref{def:hDGS}.) Second, we show a simple test that can distinguish between the discrete Gaussian distribution with parameter slightly above smoothing and the discrete Gaussian with parameter below smoothing. (See Lemma~\ref{lem:smoothcovariance}.) With this, we obtain an $O(1)$-approximation algorithm for the smoothing parameter that runs in $2^{n/2+o(n)}$ time.

\paragraph{Further applications.} 
Another fundamental problem on lattices
is the Closest Vector Problem (CVP), in which we must find a closest lattice vector to some target vector $\vec{t}$. 
$\problem{CVP}$ is known to be at least as hard as $\problem{SVP}$, as
there is a polynomial-time approximation-preserving reduction from $\problem{SVP}$ to
$\problem{CVP}$~\cite{GMSS99}. Furthermore, almost all of the major lattice problems reduce to $\problem{CVP}$
in this way~\cite{Micciancio08}.

The fastest exact algorithm for $\problem{CVP}$ is again the
$\widetilde{O}(4^n)$-time and $\widetilde{O}(2^n)$-space algorithm due to
Micciancio and Voulgaris~\cite{MV13} (which in fact more directly solves $\problem{CVP}$ than
$\problem{SVP}$). For approximation factor $\gamma = 1+r$ for $r > 0$, randomized sieving techniques have
been shown capable of solving $\gamma$-approximate $\problem{CVP}$ in
$2^{O(n)}(1+1/r)^{O(n)}$ time and space~\cite{AKS02,BN09}, though little effort has
been made to optimize the constant in the exponent. 

Based on an embedding trick of Kannan~\cite{Kannan87} and standard concentration bounds on the discrete Gaussian, 
we show how to use our sampler to solve $1.97$-approximate \problem{CVP} in time $2^{n+o(n)}$. As mentioned above, the reductions of~\cite{Micciancio08}
show that this yields the same approximation factor and running time for almost all lattice problems.  

 Also, our algorithm
from Theorem~\ref{thm:smooth-dgs-inf} gives $2^{n/2 + o(n)}$-time algorithms for 
$.422$-\problem{BDD} (Corollary~\ref{cor:bdd}) and 
$O(\sqrt{n \log n})$-approximate \problem{SIVP} (Corollary~\ref{cor:sivp}).

\subsection{Conclusions and open problems} 

Our work raises many questions and potential avenues for improvement. 
Firstly,
we suspect that the algorithm from Theorem~\ref{thm:smooth-dgs-inf} can be modified to work for an arbitrary parameter $s$ with the same running time of roughly $2^{n/2}$ (at least to sample a single vector). Such a result would subsume 
Theorem~\ref{thm:gen-dgs-inf} and would lead to an improved algorithm for \problem{SVP}, as well as 
other problems. We have made some modest progress towards proving this, but a solution still
seems far. 

Another central open problem is whether $\problem{SVP}$ can be solved in singly exponential time but only polynomial space. The best running time known for polynomial-space algorithms is the $n^{O(n)}$ obtained by enumeration-based methods~\cite{Kannan87,Helfrich86,HanrotStehle07,MicciancioWalter14}.

Finally, this work shows that Discrete Gaussian Sampling is a lattice problem of central importance. However, \problem{DGS} for parameters below smoothing is not nearly as well-understood as many other lattice problems, and many natural questions remain open. For example, is there a dimension-preserving reduction from \problem{DGS} to \problem{CVP}? (This question was answered by \cite{DGStoSVP} after a preliminary version of this work appeared.) Is (centered) \problem{DGS} NP-hard?

\paragraph{Follow-up work.}
In a follow-up work by three of us~\cite{ADS15} we generalize Theorem~\ref{thm:gen-dgs-inf} by presenting a $2^{n+o(n)}$-time algorithm to sample from the \emph{shifted} discrete Gaussian $D_{\lat - \vec{t}, s}$ for any $\vec{t} \in \R^n$ and $s>0$. 
As an application of that algorithm, we show in~\cite{ADS15} how to obtain a $2^{n+o(n)}$-time algorithm for \emph{exact} \problem{CVP}
(which is a harder problem than $\problem{SVP}$, as follows from the dimension-preserving reduction in~\cite{GMSS99}). 
While those results are stronger than some of the results presented in this paper, the 
proofs in~\cite{ADS15} are also significantly more involved.

\paragraph{Organization.} In Section~\ref{sec:prelims}, we overview the
necessary background material and give the basic definitions used throughout
the paper. In Section~\ref{sec:main-dgs}, we give our general $2^{n+o(n)}$-time DGS
sampler (Theorem~\ref{thm:DGS}). In Section~\ref{sec:SVP}, we prove
our bound on the number of discrete Gaussian samples needed for $\problem{SVP}$
(Lemma~\ref{lem:Gaussian-sum-bound} and Proposition~\ref{prop:boundonmass-KL}) and give our reduction from \problem{SVP} to \problem{DGS} (Theorem~\ref{thm:svptodgs}). In Section~\ref{sec:abovesmooth}, we
give our $2^{n/2+o(n)}$-time DGS sampler for parameters above smoothing
(Theorem~\ref{thm:smoothDGS}). In Section~\ref{sec:gapSVP}, we show our reduction from \problem{GapSVP} to \problem{DGS} above smoothing (Theorem~\ref{thm:gapSVP}). Finally, in Section~\ref{sec:other}, we show our
$2^{n+o(n)}$-time algorithm for $1.97$-approximate $\problem{CVP}$
(Theorem~\ref{thm:cvptodgs}) and our $2^{n/2 + o(n)}$-time algorithms 
for 
$.422$-\problem{BDD} (Corollary~\ref{cor:bdd}) and 
$O(\sqrt{n \log n})$-approximate \problem{SIVP} (Corollary~\ref{cor:sivp}).

\section{Preliminaries}
\label{sec:prelims}

Let $\N = \{0,1,\ldots, \}$. Except where we specify otherwise, we use $C$, $C_1$, and $C_2$ to denote universal positive constants, which might differ from one occurrence to the next. We use bold letters $\vec{x}$ for vectors and denote a vector's coordinates with indices $x_i$. Throughout the paper, $n$ will always be the dimension of the ambient space $\R^n$.

\subsection{Lattices}

A rank $d$ lattice $\lat\subset \R^n$ is the set of all integer linear combinations of $d$ linearly independent vectors $\basis = (\vec{b}_1, \ldots, \vec{b}_d )$. $\basis$ is called a basis of the lattice and is not unique. Formally, a lattice is represented by a basis $\basis$ for computational purposes, though for simplicity we often do not make this explicit.  If $n = d$, we say that the lattice has full rank, and we often assume this as results for full-rank lattices naturally imply results for arbitrary lattices. 

Given a basis, $(\vec{b}_1,\ldots, \vec{b}_d)$, we write $\lat(\vec{b}_1,\ldots, \vec{b}_d)$ to denote the lattice with basis $(\vec{b}_1,\ldots, \vec{b}_d)$. The length of a shortest non-zero vector in the lattice is written $\lambda_1(\lat)$. For a vector $\vec{t} \in \R^n$, we write $\dist(\vec{t}, \lat)$ to denote the distance between $\vec{t}$ and the lattice, $\min_{\vec{y} \in \lat}(\length{\vec{t} - \vec{y}})$.

For a lattice $\lat \subset \R^n$, the dual lattice, denoted $\lat^*$, is defined as the set of all points in $\spn (\lat)$ that have integer inner products with all lattice points,
\[ \lat^* = \{ \vec{w} \in \spn(\lat) : \forall \vec{y} \in \lat, \inner{\vec{w},\vec{y}} \in \Z \}\;. \]
Similarly, for a lattice basis $\basis = (\vec{b}_1,\ldots, \vec{b}_d)$, we define the dual basis $\basis^*=(\vec{b}_1^*,\ldots, \vec{b}_d^*)$ to be the unique set of vectors in $\spn(\lat)$ satisfying $\inner{ \vec{b}_i^*, \vec{b}_j} = \delta_{i,j} $. It is easy to show that $\lat^*$ is itself a rank $d$ lattice and $\basis^*$ is a basis of $\lat^*$.

\begin{definition}
For a lattice $\lat$, the $i$th successive minimum of $\lat$ is
\[ \lambda_i(\lat) = \inf \{ r : \dim (\spn (\lat \cap B(\vec0, r))) \geq i \}  \;.\]
\end{definition}

In other words, the $i$th successive minimum of $\lat$ is the smallest value $r$ such that there are $i$ linearly independent vectors in $\lat$ of length at most $r$. 

\subsection{The discrete Gaussian distribution}

For any $s>0$, we define the function $\rho_s : \R^n \rightarrow\R$ as $\rho_s(\vec{t}) = \exp(-\pi \length{\vec{t}}^2/s^2)$. When $s=1$, we simply write $\rho(\vec{t})$. For a discrete set $A \subset \R^n$ we define $\rho_s(A)=\sum_{\vec{x}\in A} \rho_s(\vec{x})$. 
\begin{definition} 
For a lattice $\lat \subset \R^n$ and a vector $\vec{t} \in \R^n$, let $D_{\lat + \vec{t},s}$ be the probability distribution over $\lat + \vec{t}$ such that the probability of drawing $\vec{x} \in \lat + \vec{t}$ is proportional to $\rho_{s}(\vec{x})$. We call this the discrete Gaussian distribution over $\lat + \vec{t}$ with parameter $s$.
\end{definition}

We make frequent use of the discrete Gaussian over the cosets of a sublattice. If $\lat' \subseteq \lat$ is a sublattice of $\lat$, then the set of cosets, $\lat/\lat'$ is the set of translations of $\lat'$ by lattice vectors, $\vec{c} = \lat' + \vec{y}$ for some $\vec{y} \in \lat$. It is easily seen from the Poisson summation formula that for any $\coset \in \lat/\lat'$, $\rho_s(\lat') \geq \rho_s(\coset)$, i.e., the zero coset has maximal weight (see, e.g., \cite{banaszczyk}). We use this fact throughout the paper. In particular, it follows that $\rho_s(\lat)/\rho_s(\lat') \leq |\lat/\lat'|$.

Banaszczyk proved the following two bounds on the discrete Gaussian~\cite{banaszczyk}.

\begin{lemma}[{\cite[Lemma 1.4]{banaszczyk}}]
\label{lem:banaszczyk} 
For any lattice $\lat\subset\R^n$ and $s > 1$,
\[
\rho_s(\lat) \leq s^n \rho(\lat)
\;.
\]
\end{lemma}

\begin{lemma}[{\cite[Lemma 2.13]{cvpp}}]
\label{lem:banaszczyktail} 
For any lattice $\lat\subset\R^n$, $s > 0$, $\vec{u} \in \R^n$, and $t \geq1/\sqrt{2\pi}$,
\[
\Pr_{\vec{X} \sim D_{\lat + \vec{u}, s}}[\length{\vec{X}} > t s\sqrt{n} ] < \frac{\rho_s(\lat)}{\rho_s(\lat+ \vec{u})}\big( \sqrt{2 \pi e t^2} \exp(-\pi t^2) \big)^n
\; .
\]

\end{lemma}

\begin{definition}
\label{def:smooth}
For a lattice $\lat \subset \R^n$ and $\epsilon > 0$, we define the smoothing parameter $\eta_\epsilon(\lat)$ as the unique value satisfying $\rho_{1/\eta_\epsilon(\lat)}(\lat^* \setminus \{ \vec0 \}) = \epsilon $.
\end{definition}

We note that if $\lat' \subseteq \lat$, then $\eta_\eps(\lat) \leq \eta_\eps(\lat')$, and we have $\eta_\eps(s\lat) = s \eta_\eps(\lat)$. The name smoothing parameter comes from the following fact.

\begin{claim}
\label{clm:smooth}
For any lattice $\lat \subset \R^n$ and $\eps \in (0,1)$, if $s \geq \eta_{\eps}(\lat)$, then
for all $\vec{t} \in \R^n$,
\[\frac{\rho_s(\lat + \vec{t})}{\rho_s(\lat)} \geq \frac{1-\eps}{1+\eps}
\; .
\]
\end{claim}

Finally, we will need the following basic bounds on the smoothing parameter, the first of which is essentially the same as \cite[Lemma 2.4]{CDLP13}.

\begin{lemma}
\label{lem:doublesmooth}
For any lattice $\lat \subset \R^n$, $\eps \in (0,1)$, and $k > 1$, we have
$k\eta_{\eps}(\lat) > \eta_{\eps^{k^2}}(\lat)
$.
\end{lemma}
\begin{proof}
Suppose without loss of generality that $\eta_{\eps}(\lat) = 1$. Then,
\begin{align*}
\rho_{1/k}(\lat^* \setminus \{\vec0\}) &= \sum_{\vec{y} \in \lat^* \setminus \{ \vec0\}} \rho(\lat)^{k^2}
< \Big( \sum_{\vec{y} \in \lat^* \setminus \{ \vec0\}} \rho(\lat)\Big)^{k^2} = \eps^{k^2}
\; .
\qedhere
\end{align*}
\end{proof}

\begin{lemma}\label{lem:etalambdansqrt} For any lattice $\lat \subset \R^n$ and $\eps = 0.99$,
\[
\eta_{\epsilon}(\lat) \geq C\frac{\lambda_n(\lat)}{\sqrt{n}} \; .\]
\end{lemma}
\begin{proof} 
Suppose $\lambda_n(\lat) > 500\sqrt{n}\eta_{\epsilon}(\lat)$. Then there exists a $\vec{u} \in \R^n$ such that $\dist(\vec{u},\lat) > 250\sqrt{n}\eta_{\epsilon}(\lat)$. Then, using Lemma~\ref{lem:banaszczyktail},
\[
\rho_{\eta_{\epsilon}(\lat)}(\lat + \vec{u})
= \rho_{\eta_{\epsilon}(\lat)}((\lat+\vec{u})\setminus  B(\vec{0}, 250\sqrt{n}\eta_{\epsilon}(\lat)))
\leq 200^{-n} \rho_{\eta_{\epsilon}(\lat)}(\lat).
\]
Using Claim~\ref{clm:smooth}, this gives
\[
\eta_{\epsilon}(\lat)^n\det(\lat^*)(1-\epsilon) \leq 200^{-n} \eta_{\epsilon}(\lat)^n\det(\lat^*)(1+\epsilon),
\]
which is a contradiction.
\end{proof}

\begin{lemma}[{\cite[Lemma 2.8]{MP12}}]
\label{lem:subgaussianity}
For any lattice $\lat \subset \R^n$, $s > 0$, $t > 0$, and unit vector $\vec{v} \in \R^n$,
\[
\Pr_{\vec{X} \sim D_{\lat, s}}[\abs{\inner{\vec{X}, \vec{v}}} \geq t ] \leq 2 e^{-\pi t^2/s^2}
\; .
\]
\end{lemma}

\subsection{The Gram-Schmidt orthogonalization}

Given a basis, $\basis = (\vec{b}_1,\ldots, \vec{b}_n)$, we define its Gram-Schmidt orthogonalization $(\gs{\vec{b}}_1,\ldots, \gs{\vec{b}}_n)$ by
\[  \gs{\vec{b}}_i = \pi_{\{b_1, \ldots, b_{i-1} \}^\perp}(\vec{b}_i) \; . \]
Here, $\pi_A$ is the orthogonal projection on the subspace $A$ and $\{b_1, \ldots, b_{i-1} \}^\perp$ denotes the subspace orthogonal to $b_1, \ldots, b_{i-1}$.

\subsection{Lattice problems}

The following problem plays a central role in this paper.

\begin{definition}
\label{def:dgs}
For $\eps = \eps(n) \geq 0$, $\sigma$ a function that maps lattices to non-negative real numbers, and $m = m(n) \in \N$, $\DGS{\eps}{\sigma}{m}$ (the Discrete Gaussian Sampling problem) is defined as follows: 
The input is a basis $\basis$ for a lattice $\lat \subset \R^n$ and a parameter $s > \sigma(\lat)$. The goal is to output a sequence of $m$ vectors whose joint distribution is $\eps$-close to $D_{\lat, s}^m$.
\end{definition}

We omit the parameter $\eps$ if $\eps = 0$, the parameter $\sigma $ if $\sigma = 0$, and the parameter $m$ if $m = 1$. We stress that $\eps$ bounds the statistical distance between the \emph{joint} distribution of the output vectors and $m$ independent samples of $D_{\lat,s}$. 

For our applications, we consider the following lattice problems.

\begin{definition}
The search problem $\problem{SVP}$ (Shortest Vector Problem) is defined as follows: The input is a basis $\basis$ for a lattice $\lat \subset \R^n$. The goal is to output a vector $\vec{y} \in \lat$ with $\length{\vec{y}} = \lambda_1(\lat)$.
\end{definition}

\begin{definition}
For $\gamma = \gamma(n) \geq 1$ (the approximation factor), the decision problem $\gamma\text{-}\problem{GapSVP}$ is defined as follows: The input is a basis $\basis$ for a lattice $\lat \subset \R^n$ and a number $d > 0$. The goal is to output yes if $\lambda_1(\lat) < d$ and no if $\lambda_1(\lat) \geq \gamma \cdot d$.
\end{definition}

\begin{definition}
For $\gamma = \gamma(n) \geq 1$ (the approximation factor), the search problem $\gamma\text{-}\problem{CVP}$ (Closest Vector Problem) is defined as follows: The input is a basis $\basis$ for a lattice $\lat \subset \R^n$ and a target vector $\vec{t} \in \R^n$. The goal is to output a vector $\vec{y} \in \lat $ with $\length{\vec{y} - \vec{t}} \leq \gamma \cdot \dist(\vec{t}, \lat)$.
\end{definition}

\begin{definition}
For $\alpha = \alpha(n) < 1/2$ (the approximation factor), the search problem $\alpha\text{-}\problem{BDD}$ (Bounded Distance Decoding) is defined as follows: The input is a basis $\basis$ for a lattice $\lat \subset \R^n$ and a target vector $\vec{t} \in \R^n$ with $\dist(\vec{t}, \lat) \leq \alpha \cdot \lambda_1(\lat)$. The goal is to output a closest lattice vector to $\vec{t}$.
\end{definition}

Note that, while our other problems become more difficult as the approximation factor $\gamma$ becomes smaller, $\alpha\text{-}\problem{BDD}$ becomes more difficult as $\alpha$ gets larger. For convenience, when we discuss the running time of algorithms solving the above problems, we ignore polynomial factors in the bit-length of the individual input basis vectors. (I.e., we consider only the dependence on the ambient dimension $n$.)

\subsection{Some lattice algorithms}

The following theorem was proven by Ajtai, Kumar, and Sivakumar~\cite{AKS01}, building on work of Schnorr~\cite{Schnorr87}.
(While we use the AKS algorithm repeatedly in the sequel for convenience, we note that we could instead use the conceptually simpler algorithm from~\cite{Schnorr87} to obtain asymptotically identical results.)
\begin{theorem}
\label{thm:BKZ}
There is an algorithm that takes as input a lattice $\lat \subset \R^n$ and $r \geq 2$ and outputs an $r^{n/r}$-reduced basis of $\lat$ in time $\exp(O(r)) \cdot \poly(n)$, where we say that a basis $\basis = (\vec{b}_1,\ldots, \vec{b}_n)$ of a lattice $\lat$ is $\gamma$-reduced for some $\gamma \geq 1$ if 
\begin{enumerate}
\item $\length{\vec{b}_1} \leq \gamma \cdot \lambda_1(\lat)$; and
\item $\pi_{\{ \vec{b_1} \}^\perp}(\vec{b}_2), \ldots, \pi_{\{ \vec{b}_1 \}^\perp}(\vec{b}_n)$ is a $\gamma$-reduced basis of $\pi_{\{ \vec{b_1} \}^\perp}(\lat)$.
\end{enumerate}
\end{theorem}

In order to initialize our algorithm, we will need to use a Gaussian sampler such as the one 
given by Gentry, Peikert, and Vaikuntanathan~\cite{GPV08}.
For convenience, we use the following 
modest strengthening of this result, which provides exact samples and gives slightly better bounds on the parameter $s$.
\begin{theorem}[{\cite[Lemma 2.3]{BLPRS13}}]
\label{thm:GPV}
There is a probabilistic polynomial-time algorithm that takes as input a basis $\basis $ for a lattice $\lat \subset \R^n$ and $s > \length{\gs{\basis}}\sqrt{C \log n}$ and outputs a vector that is distributed exactly as $D_{\lat, s}$, where $\length{\gs{\basis}} := \max \length{\gs{\vec{b}}_i}$.
\end{theorem}

Ideally, we would like to use Theorem~\ref{thm:BKZ} and Theorem~\ref{thm:GPV} to solve $\problem{DGS}_\sigma$ for $\sigma = \lambda_1(\lat) \cdot r^{n/r}$ in time $\approx \exp(O(r))$. Unfortunately, this does not work. The problem is that Theorem~\ref{thm:GPV} only allows us to sample from $D_{\lat, s}$ if \emph{all} of the Gram-Schmidt vectors are smaller than $\approx s$. We cannot hope to achieve this even for $s \approx \lambda_1(\lat) \cdot r^{n/r}$. Indeed, there may not even be such a basis! Instead, we show that we can sample from a sublattice for which we can find such a basis, and we show that this sublattice contains all of the \scarequotes{short} lattice points.

\begin{proposition}
\label{prop:startgauss}
There is an algorithm that takes as input a lattice $\lat \subset \R^n$ with 
$n \geq 2$, $2 \leq r \leq O(n)$, $M \in \N$ (the desired number of output vectors), and 
$s > 0$
 and outputs a sublattice $\smalllat$ and $M$ independent samples from $D_{\smalllat, s}$ in time $(2^{O(r)} + M) \cdot \poly(n)$. The sublattice $\smalllat $ contains all vectors in $\lat$ of length at most $r^{-n/r}s$.
Furthermore, if 
\[
s 
\geq (Cr)^{n/r} \cdot \sqrt{n \log n} \cdot \eta_{0.99}(\lat)
\; ,
\] 
then $\smalllat = \lat$.
\end{proposition}
\begin{proof}
On input $\lat$ the algorithm first runs the procedure from Theorem~\ref{thm:BKZ} on $\lat$ with parameter $C_1 r$, receiving output $\basis = (\vec{b}_1,\ldots, \vec{b}_n)$. 
Let $\gs{\vec{b}}_1,\ldots, \gs{\vec{b}}_n$ be the corresponding Gram-Schmidt vectors, and let $k$ be maximal such that $\length{\gs{\vec{b}}_i} \leq s/\sqrt{C_2\log n}$ for all $i \leq k$. The algorithm then runs the procedure from Theorem~\ref{thm:GPV} $M$ times on input $(\vec{b}_1,\ldots, \vec{b}_k)$ and $s $ and outputs the result together with $\smalllat = \lat(\vec{b}_1,\ldots, \vec{b}_k)$.

The running time is clear. It follows immediately from Theorem~\ref{thm:GPV} that the output has the correct distribution. If $k = n$, then we are done. Otherwise, let $\biglat = \pi_{\smalllat^\perp}(\lat)$. By Theorem~\ref{thm:BKZ}, we have that 
 \[
 \lambda_1(\biglat) 
 \geq (C_1r)^{-n/(C_1r)} \length{\gs{\vec{b}}_{k+1}} 
 > (C_1r)^{-n/(C_1r)}s/\sqrt{C_2 \log n} > r^{-n/r} s
 \; .
 \]
The main result follows by noting that $\vec{y} \in \lat \setminus \smalllat$ implies that $\length{\vec{y}} \geq \lambda_1(\biglat) > r^{-n/r}s$.

Finally, we note that $\length{\gs{\vec{b}}_i} \leq (C_1 r)^{n/(C_1r)} \cdot  \lambda_n(\lat)$ for all $i$. It follows that, if $s \geq (Cr)^{n/r} \sqrt{\log n} \cdot \lambda_n(\lat)$, then $k = n$, and therefore $\smalllat = \lat$. The second statement then follows from Lemma~\ref{lem:etalambdansqrt}.
\end{proof}

\subsection{Probability distributions}

\begin{definition}[Poisson distribution]
The Poisson distribution with parameter $\lambda > 0$ is the distribution defined by 
\[
\Pr_{X \sim \Pois(\lambda)}[X = r] = \frac{\lambda^r}{r!}\cdot e^{-\lambda}
\] 
for all $m \in \N$.
\end{definition}

Intuitively, the Poisson distribution is the distribution obtained by, e.g., counting the number of decay events over some fixed time period in some large, homogenous radioactive source. The parameter $\lambda$ is just the expected count.

\begin{lemma}[Poisson tail bounds~\cite{Glynn87}]
\label{lem:glynn}
For $\lambda > 0$ let $X$ be a $\pois(\lambda)$ random variable. Then,
\begin{itemize}
\item
for any $0 \le m < \lambda$,
\[
\Pr(X \le m ) \le \frac{\exp(-\lambda) \lambda^m}{m! ( 1-(m/\lambda))} \; ,
\]
\item
and for any $m > \lambda - 1$,
\[
\Pr(X \ge m) \le \frac{\exp(-\lambda) \lambda^m}{m! ( 1-(\lambda/(m+1)))}
\; .
\]
\end{itemize}
\end{lemma}

\begin{corollary}\label{cor:poistail}
For any $\alpha>0$, there exist $C_1, C_2>0$ such that the following holds for all $m \ge 1$. 
If $X$ is a $\pois(\lambda)$ random variable for some $\lambda < (1-\alpha)m$ then
\[
\Pr(X \ge m) \le C_1 \exp(-C_2 m) \; ,
\]
and similarly, if $\lambda > (1+\alpha)m$ then
\[
\Pr(X \le m) \le C_1 \exp(-C_2 m) \; .
\]
\end{corollary}
\begin{proof}
Stirling's approximation implies the inequality $m! \ge (m/e)^m$ valid for all $m\ge 1$, which together with Lemma~\ref{lem:glynn} implies in both cases the upper bound 
\[
C \exp(- m (\lambda/m - \log(\lambda/m)+1))\; .
\]
The function $x - \log x + 1$ is non-negative and strictly convex on $x > 0$ and obtains its minimum of $0$ at $x=1$. As a result, it is uniformly bounded away from $0$ for all $x$ satisfying $|x-1| \ge \alpha$.
\end{proof}

We will also need the Chernoff-Hoeffding bound~\cite{hoeffding}.

\begin{lemma}[Chernoff-Hoeffding bound]
\label{lem:chernoff}
Let $X_1, \ldots, X_N $ be independent and identically distributed random variables with $0 \leq X_i \leq a$ and expectation $\mu$. 
Then, for any $\delta > 0$, 
\[
\Pr\Big[\frac{1}{N}\cdot\sum X_i \geq (1+\delta) \mu \Big] \leq \exp(-C \delta^2 N\mu/a)
\; ,
\]
and
\[
\Pr\Big[\frac{1}{N}\cdot\sum X_i \leq (1-\delta)\mu \Big] \leq \exp(-C \delta^2 N\mu/a)
\; .
\]
\end{lemma}

\begin{lemma}[Multinomial to independent Poisson]
\label{lem:independence}
Let $\lambda > 0$ and $\vec{p} \in [0,1]^N$ with $\sum p_i = 1$. Consider the process that first samples $r \sim \Pois(\lambda)$ and then samples $X_1,\ldots, X_r$ independently with $\Pr[X_j = i] = p_i$. For each $i$, let $Y_i$ be the number of occurrences of $i$ in the sequence $X_1,\ldots, X_r$. Then, $Y_i$ is distributed as $\pois(\lambda p_i)$ independently of the other $Y_j$.
\end{lemma}
\begin{proof}
Considering the joint distribution, we have
\begin{align*}
\Pr[\vec{Y} = \vec{a}] &= \Pr[r = \length{\vec{a}}_1] \cdot \Pr[\vec{Y} = \vec{a} | r = \length{\vec{a}}_1]\\ 
&=  \lambda^{\length{\vec{a}}_1} e^{-\lambda} \prod_i \frac{p_i^{a_i}}{a_i!} \\
&= \prod_i \left(\frac{(\lambda p_i)^{a_i}}{a_i!} \cdot e^{- \lambda p_i} \right)
\; ,
\end{align*}
as needed.
\end{proof}

\begin{claim}[Poisson to Bernoulli]
\label{clm:poissontobinom}
For $\lambda \le 1$ and $\kappa \ge 2$, consider the procedure obtained by sampling $r$ from $\pois(\lambda)$ and then outputting $1$ with probability $\min\{1,\ r/\kappa\}$ and $0$ otherwise. The output of this procedure is within statistical distance $1/(\floor{\kappa}!)$ of the Bernoulli distribution $\bern(\lambda/\kappa)$.
\end{claim}
\begin{proof}
If $X$ is distributed like $\pois(\lambda)$, the statistical distance is given by
\begin{align*}
\expect[X/\kappa - \min\{1,\ X/\kappa\}] &= 
\expect[\max\{0,\ X/\kappa - 1 \}] \\
&\le  \kappa^{-1} \expect[1_{X > \kappa} \cdot X] \\
&= \kappa^{-1} \sum_{r = \floor{\kappa}+1}^\infty r \lambda^r \exp(-\lambda) / r!  \\
&= \kappa^{-1} \lambda \sum_{r = \floor{\kappa}}^\infty \lambda^r \exp(-\lambda) / r! \; ,
\end{align*}
which is at most $1/(\floor{\kappa}!)$ by Lemma~\ref{lem:glynn} and our choice of parameters.
\end{proof}

\section{Sampling from the discrete Gaussian}
\label{sec:main-dgs}

\subsection{Sampling from the square}
\label{sec:squaresampler}

Recall that a naive bucketing procedure does not weight cosets in the way that we would like. In particular, the resulting number of vectors in the cosets is distributed with probabilities proportional to $\rho_s(\coset)$, while we would like the probabilities to be proportional to $\rho_s(\coset)^2$. Theorem~\ref{thm:squaresampler} shows how to use samples from any multinomial distribution to sample from the \scarequotes{squared distribution} (with small error).

The \scarequotes{square sampler} that we present in Theorem~\ref{thm:squaresampler} needs to compute an estimate of the maximal probability $\max p_i$, given samples from some probability distribution with respective probabilities $(p_1,\ldots, p_N)$. 
The following proposition shows that there is a relatively efficient way of estimating $\max p_i$.
The proposition is included here for completeness. In our application, we will know which of the elements $1,\ldots, N$ has maximal probability, so we could instead simply estimate $\max p_i$ directly.

\begin{proposition}[Estimating $\pmax$]
\label{prop:estimatinglinfty}
There is an algorithm that takes as input $\kappa \geq 1$ (the confidence parameter) and a sequence of $M$ elements from $\{1, \ldots, N \}$ and outputs a value $\approxpmax$ such that, if the input consists of $M \geq \kappa/\pmax $ independent samples from the distribution that assigns probability $p_i$ to element $i$, then 
\[ \pmax \leq \approxpmax \leq 4\pmax
\] except with probability at most $C_1 N \log N \exp(-C_2 \kappa)$, where $\pmax = \max p_i$ . The algorithm runs in time $M \cdot \poly(\log \kappa, \log N)$.
\end{proposition}

\begin{proof}
The algorithm is the following. Initialize $p=1$. Sample $r$ from $\pois(\kappa/p)$ and read the next $r$ elements in the input sequence (or fail if there are not enough elements remaining). 
Count how many times each $i \in \{1,\ldots,N\}$ appears in this subsequence. If there exists an $i$ appearing at least $\kappa/3$ times, output $p$ and stop. Otherwise, divide $p$ by $2$ and repeat. 

The running time is clear. By Lemma~\ref{lem:independence}, at each iteration the number of times $i$ appears is distributed like $\pois(\kappa p_i / p)$ independently of everything else. 
Consider now the iterations with $p>4\pmax$. 
Since $\pmax > 1/N$, there are at most $O(\log N)$ such iterations, and in each there are $N$ possible values of $i$. 
Therefore, by Corollary~\ref{cor:poistail} and a union bound, the probability that there exists an iteration with $p>4\pmax$ and an $i$ that appears 
there at least $\kappa/3$ times is at most $C_1 N \log N \cdot \exp(-C_2 \kappa)$. 
Finally, consider the iteration in which $\pmax < p \le 2 \pmax$, and let $i$ be the index achieving $\pmax$. 
Then by Corollary~\ref{cor:poistail} again, with all but probability $C_1 \exp(-C_2 \kappa)$, the item $i$ appears at least $\kappa/3$ times. To summarize, assuming none of the bad events happens, 
the output satisfies $\pmax < \approxpmax \le 4 \pmax$ as desired. 
\end{proof}

\begin{definition}
For a vector $\vec{p} \in [0,1]^N$ with $\sum p_i = 1$, let $\vec{p}^2 = (p_1^2/\pcol, \ldots, p_N^2/\pcol)$ where $\pcol := \sum p_i^2$.
\end{definition}

\begin{theorem}[Square sampler]
\label{thm:squaresampler}
There is an algorithm that takes as input $\kappa \geq 2$ (the confidence parameter)
and $M$ elements from $\{1, \ldots, N \}$ and outputs a sequence of elements from the same set such that
\begin{enumerate}
\item \label{item:squareruntime} the running time is $M \cdot \poly(\log \kappa, \log N)$;
\item \label{item:squareinputoutput} each $i \in \{1,\ldots, N\}$ appears at least twice as often in the input as in the output; and
\item  \label{item:squaredistribution} if the input consists of 
$M \geq 10 \kappa^2/\max p_i$ 
independent samples from the distribution that assigns probability $p_i$ to element $i$, then
the output is within statistical distance $C_1 M N \log N \exp(-C_2\kappa)$ of $m$ independent samples with respective probabilities $\vec{p}^2$ where $m \geq  M\cdot \sum p_i^2/(32\kappa \max p_i)$ is a random variable.
\end{enumerate}
\end{theorem}
\begin{proof}
The algorithm first runs the procedure from Proposition~\ref{prop:estimatinglinfty} on the first $M/2$ elements from its input sequence, receiving as output $\approxpmax$. The algorithm then reads the remaining elements in sequence.
If it ever reads the last element of the input, it fails.
For $j=1, \ldots, M\approxpmax/4$, the algorithm samples $r$ according to $ \pois(1/\approxpmax)$ and takes the next $r$ unused elements in the input. 
For $i = 1,\ldots, N$, let $a_{i,j}$ be the number of times element $i$ appears in the $j$th such subsequence. For each $i,j$ let $b_{i,j}$ be $1$ with probability $\min\{ 1,\, a_{i,j}/\kappa \}$ and $0$ otherwise. 
(To achieve the correct running time, we do not actually explicitly store these values when $a_{i,j} = b_{i,j} = 0$.) 

Finally, the algorithm looks through the next $M/6$ elements, one element at a time (or it fails if there are not $M/6$ elements remaining). When it sees element $i$, it adds it to its output if $b_{i,j} = 1$ where $j \geq 1$ is the smallest index such that $b_{i,j}$ is unused (or it fails if there is no unused $b_{i,j}$).

The running time of the algorithm is clear. We first prove that the output elements have the correct distribution, ignoring failure. By Proposition~\ref{prop:estimatinglinfty}, we can assume that $\max p_i \leq \approxpmax \leq 4\max p_i$, introducing statistical distance at most $C_1 N \log N \exp(-C_2 \kappa)$. Then, by Lemma~\ref{lem:independence}, the $a_{i,j}$ are distributed independently as $\pois(p_i/\approxpmax)$. 
By Claim~\ref{clm:poissontobinom}, each $b_{i,j}$ is within statistical distance $C_1 \exp(-C_2\kappa)$ of $\bern(p_i/(\kappa \approxpmax))$. 
So, we assume that the $b_{i,j}$ are independently distributed exactly as $\bern(p_i/(\kappa \approxpmax))$, introducing statistical distance that is at most $C_1 NM\exp(-C_2\kappa)$. 
Then, in the final stage of the algorithm, the probability of outputting $i$ at each step is $p_i^2/(\kappa \approxpmax)$. 
Hence, the individual output samples have the correct distribution. 
By the Chernoff-Hoeffding bound (Lemma~\ref{lem:chernoff}), the size of the output will be at least $M\sum p_i^2/(8\kappa \approxpmax) \geq M\sum p_i^2/(32 \kappa \max p_i)$ except with probability at most $\exp(-C \kappa)$.

We now prove that the algorithm rarely fails. The number of inputs used in the first stage is distributed as $\pois(M/4)$, which by Corollary~\ref{cor:poistail} is at most 
$M/2$ except with probability at most $C_1\exp(-C_2\kappa)$. 
Applying the Chernoff-Hoeffding bound again, we have that the number of coins $b_{i,j}$ used for a fixed $i$ is at most $M p_i/4 \leq M\approxpmax/4$ except with probability at most $\exp(-C\kappa)$. So, the algorithm fails with probability at most $C_1N\exp(-C_2\kappa)$.

Finally, we note that each $i$ appearing in the output corresponds to two copies of $i$ appearing in the input: one corresponding to some value $a_{i,j} > 0$ and another sampled in the final stage.
\end{proof}

\subsection{A discrete Gaussian combiner}
\label{sec:combiner}
Ideally, we would like the average of two vectors sampled from $D_{\lat,s}$ to be distributed as $D_{\lat, s'}$ for some $s' < s$. Unfortunately, this is false for the simple reason that the average of two lattice vectors may not be in the lattice! 
The following lemma shows that we do obtain the desired distribution if we condition on the result being in the lattice. The number of vectors that we output will depend on the expression $\sum \rho_s(\coset)^2$ where $\coset$ ranges over all cosets of $2\lat $ over $\lat$, so we analyze this expression as well. 
(Note that, for two lattice vectors $\vec{X}_1$ and $\vec{X}_2$, we have $(\vec{X}_1 + \vec{X}_2)/2 \in \lat$ if and only if $\vec{X}_1$ and $\vec{X}_2$ are in the same coset over $2\lat$. So, the cosets of $2\lat$ arise naturally in this context.)

\begin{lemma}
\label{lem:sumofgaussians}
Let $\lat \subset \R^n$ and $s > 0$. Then for all $\vec{y} \in \lat$,
\begin{align}\label{eq:sumofgaussians}
\Pr_{(\vec{X}_1, \vec{X}_2) \sim D_{\lat, s}^2}[(\vec{X}_1 + \vec{X}_2)/2 = \vec{y} ~|~ \vec{X}_1 + \vec{X}_2 \in 2\lat] 
= \Pr_{\vec{X} \sim D_{\lat, s/\sqrt{2}}}[\vec{X} = \vec{y}]
\; .
\end{align}
Furthermore,
\[
\sum_{\coset \in \lat/(2\lat)}\rho_s(\coset)^2 = \rho_{s/\sqrt{2}}(\lat)^2
\; .
\]
\end{lemma}
\begin{proof}
Multiplying the left-hand side of~\eqref{eq:sumofgaussians} by 
$\Pr_{(\vec{X}_1, \vec{X}_2) \sim D_{\lat, s}^2}[\vec{X}_1 + \vec{X}_2 \in 2\lat]  = \rho_s(\lat)^{-2} \sum_{\coset \in \lat/(2\lat)}\rho_s(\coset)^2$ 
we get for any $\vec{y} \in \lat$,
\begin{align*}
\Pr_{(\vec{X}_1, \vec{X}_2) \sim D_{\lat, s}^2}[(\vec{X}_1 + \vec{X}_2)/2 = \vec{y}]  
&= \frac{1}{\rho_s(\lat)^2}\cdot\sum_{\vec{x} \in \lat} \rho_s(\vec{x}) \rho_s(2\vec{y} - \vec{x})\\
&= \frac{\rho_{s/\sqrt{2}}(\vec{y})}{\rho_s(\lat)^2}\cdot\sum_{\vec{x} \in \lat} \rho_{s/\sqrt{2}}(\vec{x} - \vec{y})\\
&= \frac{\rho_{s/\sqrt{2}}(\vec{y})}{\rho_s(\lat)^2}\cdot \rho_{s/\sqrt{2}}(\lat) \\
&= \rho_s(\lat)^{-2}\cdot \rho_{s/\sqrt{2}}(\lat)^2 \Pr_{\vec{X} \sim D_{\lat, s/\sqrt{2}}}[\vec{X} = \vec{y}]
\; .
\end{align*}
Hence both sides of~\eqref{eq:sumofgaussians} are proportional to each other. Since they are probabilities, they are actually equal. In particular, the ratio between them, $ \sum_{\coset \in \lat/(2\lat)}\rho_s(\coset)^2/\rho_{s/\sqrt{2}}(\lat)^2$, is one.
\end{proof}

\begin{proposition}
\label{prop:combiner}
There is an algorithm that takes as input a lattice $\lat \subset \R^n$, $\kappa \geq 2$ (the confidence parameter), 
and a sequence of vectors from $\lat$,
and outputs a sequence of vectors from $\lat$ such that, if the input consists of $M \geq 10 \kappa^2\cdot \rho_s(\lat)/\rho_s(2\lat)$ independent samples from $D_{\lat, s}$ for some $s > 0$, then the output is within statistical distance $M \exp(C_1 n-C_2\kappa)$ of $m$ independent samples from $D_{\lat, s/\sqrt{2}}$ where $m$ is a random variable with
\[
m \geq  M \cdot \frac{1}{32\kappa}\cdot\frac{\rho_{s/\sqrt{2}}(\lat)^2}{ \rho_s(\lat)\rho_s(2\lat)}
\; .
\] 
The running time of the algorithm is at most $M \cdot \poly(n, \log \kappa)$.
\end{proposition}
\begin{proof}
Let $(\vec{X}_1, \ldots, \vec{X}_M)$ be the input vectors. 
For each $i$, let $\coset_i \in \lat/(2\lat)$ be the coset of $\vec{X}_i$. The combiner runs the algorithm from Theorem~\ref{thm:squaresampler} with input $\kappa$
and $(\coset_1,\ldots,\coset_M)$,
receiving output $(\coset_1', \ldots, \coset_m')$. 
(Formally, we must encode the cosets as integers in $\{1,\ldots, 2^n\}$.) Finally, for each $\coset_i'$, it chooses a pair of unpaired vectors $\vec{X}_j, \vec{X}_k$ with 
$\coset_j = \coset_k = \coset_i'$ and outputs $\vec{Y}_i = (\vec{X}_j + \vec{X}_k)/2$.

The running time of the algorithm follows from Item~\ref{item:squareruntime} of Theorem~\ref{thm:squaresampler}. Furthermore, we note that by Item~\ref{item:squareinputoutput} of the same theorem, there will always be a pair of indices $j,k$ for each $i$ as above.

To prove correctness, we observe that for $\coset \in \lat/(2\lat)$ and $\vec{y} \in \coset$, 
\[ 
\Pr[\vec{X}_i = \vec{y}] = \frac{\rho_s(\coset)}{\rho_s(\lat)} \cdot \Pr_{\vec{X} \sim D_{\coset,s}}[\vec{X} = \vec{y}] \;.
\]
In particular, we have that $\Pr[\coset_i = \coset] = \rho_s(\coset)/\rho_s(\lat)$, and $2\lat$ is the coset with the highest probability. 
Then, the cosets $(\coset_1,\ldots,\coset_M)$ satisfy the conditions necessary for Item~\ref{item:squaredistribution} of Theorem~\ref{thm:squaresampler} with $\max p_i = \rho_s(2\lat)/\rho_s(\lat)$. 

Applying the theorem, up to statistical distance $M \exp(C_1 n - C_2 \kappa)$, we have that the output vectors are independent, and
\[
m 
\geq M \cdot \frac{1}{32\kappa }\cdot\frac{\sum_{\coset \in \lat/(2\lat)}\rho_s(\coset)^2}{\rho_s(\lat)\rho_s(2\lat)}
= M \cdot \frac{1}{32\kappa }\cdot\frac{\rho_{s/\sqrt{2}}(\lat)^2}{\rho_s(\lat)\rho_s(2\lat)}
\; ,
\]
where the equality follows from Lemma~\ref{lem:sumofgaussians}.
 Furthermore, we have $\Pr[\coset_i' = \coset] = \rho_s(\coset)^2/\sum_{\coset'} \rho_s(\coset')^2$ for any coset $\coset \in \lat/(2\lat)$. Therefore, for any $\vec{y} \in \lat$, 
\begin{align*}
\Pr[\vec{Y}_i = \vec{y}] &= \frac{1}{\sum \rho_s(\coset)^2} \cdot \sum_{\coset \in \lat/(2\lat)} \rho_s(\coset)^2 \cdot \Pr_{(\vec{X}_j, \vec{X}_k) \sim D_{\coset, s}^2}[(\vec{X}_j + \vec{X}_k)/2  = \vec{y}]\\
&= \Pr_{(\vec{X}_1, \vec{X}_2) \sim D_{\lat, s}^2}[(\vec{X}_1 + \vec{X}_2)/2 = \vec{y} ~|~ \vec{X}_1 + \vec{X}_2 \in 2\lat] 
\; .
\end{align*}
The result then follows from Lemma~\ref{lem:sumofgaussians}.
\end{proof}

By calling the algorithm from Proposition~\ref{prop:combiner} repeatedly, we obtain a general discrete Gaussian combiner.

\begin{corollary}
\label{cor:pipeline}
There is an algorithm that takes as input a lattice $\lat \subset \R^n$, $\ell \in \N$ (the step parameter), 
$\kappa \geq 2$ (the confidence parameter), and $M = (32\kappa)^{\ell+1} 2^n$ vectors in $\lat$ such that, if the input vectors are distributed as $D_{\lat, s}$ for some $s > 0$, then the output is a sequence of $2^{n/2}$ vectors whose distribution is within statistical distance $\ell M \exp(C_1 n-C_2\kappa )$ of independent samples from $D_{\lat, 2^{-\ell/2}s}$.
The algorithm runs in time $\ell M \cdot \poly(n,\log \kappa)$.
\end{corollary}
\begin{proof}
Let $\mathcal{X}_0 = (\vec{X}_1,\ldots, \vec{X}_M)$ be the sequence of input vectors. For $i = 0,\ldots, \ell-1$, the algorithm calls the procedure from Proposition~\ref{prop:combiner} with input $\lat$, $\kappa$, and $\mathcal{X}_i$, receiving an output sequence $\mathcal{X}_{i+1}$ of some length $M_{i+1}$. Finally, the algorithm outputs the first $2^{n/2}$ vectors of $\mathcal{X}_\ell$ (or fails if there are not enough vectors).

The running time is clear. Fix $\lat$, $s$, and $\ell$. 
For convenience, let $\psi(i) := \rho_{2^{-i/2}s}(\lat)$. Note that by Lemma~\ref{lem:banaszczyk} we have that $1\leq \psi(i)/\psi(i+1) \leq 2^{n/2}$ for all $i$, a fact that we use repeatedly below. We wish to prove by induction that $\mathcal{X}_i$ is within statistical distance $i M\exp(C_1 n-C_2\kappa )$ of $D_{\lat,2^{-i/2}s}^{M_i}$ with
\begin{align}\label{eq:inductionhypo}
M_i \geq (32\kappa)^{\ell - i + 1} \cdot 2^{n/2} \frac{\psi(i)}{\psi(i+1)}
\; 
\end{align}
for all $i$.

Since $M_0 = M = (32\kappa)^{\ell + 1} 2^{n}$ and $\psi(0)/\psi(1) \leq 2^{n/2}$, it follows that~\eqref{eq:inductionhypo} holds when $i=0$. Suppose that $\mathcal{X}_i$ has the correct distribution and~\eqref{eq:inductionhypo} holds for some $i$ with $0 \leq i < \ell$. Notice that the right-hand side of~\eqref{eq:inductionhypo} is at least 
$10 \kappa^2 \psi(i)/\psi(i+2)$ and that the latter is precisely the lower bound on 
$M_i$ appearing in Proposition~\ref{prop:combiner}.
We can therefore apply the proposition and the induction hypothesis, and obtain that (up to statistical distance at most $(i+1) M \exp(C_1 n - C_2 \kappa)$),
$\mathcal{X}_{i+1}$ has the correct distribution with
\[
M_{i+1} \geq M_{i} \cdot \frac{1}{32\kappa }\cdot \frac{\psi(i+1)^2}{\psi(i)\psi(i+2)} \geq (32\kappa)^{\ell - i} \cdot 2^{n/2} \frac{\psi(i+1)}{\psi(i+2)}
\; ,
\]
as needed. 

The result follows by noting that $M_\ell \geq 2^{n/2} \psi(\ell)/\psi(\ell+1) \geq 2^{n/2}$.
\end{proof}

\subsection{A general discrete Gaussian sampler}

\begin{theorem}
\label{thm:DGS}
There is an algorithm that solves $\DGS{\exp(-\Omega(\kappa))}{}{2^{n/2}}$  in time $2^{n +  \polylog(\kappa) + o(n)}$ for any $\kappa \geq \Omega(n)$.
\end{theorem}
\begin{proof}
On input $\lat \subset \R^n$ a lattice, $s > 0$, and $\kappa \geq \Omega(n)$, the algorithm behaves as follows. First, it runs the sampler from Proposition~\ref{prop:startgauss} on $\lat$ with parameters $r$, $\hat{s} = 2^{\ell/2} s$, and $M = (32\kappa)^{\ell + 2} \cdot 2^n$, with $r$ and $\ell$ to be set in the analysis. It receives as output a sublattice $\smalllat \subseteq \lat$ and vectors $\vec{X}_1,\ldots, \vec{X}_M \in \lat'$. It then runs the combiner from Corollary~\ref{cor:pipeline} with input $\smalllat$, $\ell $, $\kappa$, and $\vec{X}_1,\ldots, \vec{X}_M$ and outputs the result.

The running time of the first stage of the algorithm is $(2^{O(r)} + M) \cdot \poly(n)$ by Proposition~\ref{prop:startgauss}, and by Corollary~\ref{cor:pipeline}, the running time of the second stage is $M \ell \cdot \poly(n, \log \kappa)$. Setting $\ell = 4\ceil{\log \kappa + \log^2 n} $ and $r = n/\log n$, it follows that the running time is as claimed.
Applying the proposition and corollary again, we have that the output is within statistical distance $\exp(-\Omega(\kappa))$ of $D_{\smalllat, s}^{2^{n/2}}$. Furthermore, we have that $\smalllat $ contains all vectors of length at most $r^{-n/r} \hat{s} > \sqrt{\kappa} s$. 

It remains to prove that $D_{\smalllat, s}$ is within statistical distance $\exp(-\Omega(\kappa))$ of $D_{\lat, s}$. Notice that $D_{\lat', s}$ is a restriction of the distribution $D_{\lat, s}$ to $\lat'$ and hence the statistical distance between these two distributions is  
\[
\Pr_{\vec{X} \sim D_{\lat, s}}[\vec{X} \in \lat \setminus \smalllat] < \Pr_{\vec{X} \sim D_{\lat, s}}[\length{\vec{X}} > \sqrt{\kappa} s] < \exp(-\Omega(\kappa))
\; ,
\]
as needed, where we used Lemma~\ref{lem:banaszczyktail}.
\end{proof}

\section{Solving \texorpdfstring{\problem{SVP}}{SVP} in \texorpdfstring{$2^{n + o(n)}$}{2(n+o(n))} time}
\label{sec:SVP}

\subsection{A bound on the Gaussian mass}
\label{sec:boundonmass}

In this section, we prove a bound on the Gaussian mass of a lattice that follows from an upper bound on the kissing number due to Kabatjanski{\u\i} and Leven{\v{s}}te{\u\i}n~\cite{KL78}. In particular, we use the following lemma from~\cite{PS09} based on~\cite{KL78}. For convenience, we define $\beta := 2^{0.401}$, and we use this notation throughout this section.

\begin{lemma}[{{\cite[Lemma 3]{PS09}}}] 
\label{lem:KL}
Let $\lat \subseteq \mathbb{R}^n$ be a lattice with $\lambda_1(\lat) = 1$. Then for any $r \geq 1$, the number of lattice vectors of length at most $r$ is at most $\beta^{n + o(n)} r^n$.
\end{lemma} 

We now use Lemma~\ref{lem:KL} to bound $\rho_s(\lat)$. 

\begin{lemma}
\label{lem:Gaussian-sum-bound}
Let $\lat \subset \R^n$ be a lattice of rank at least one. Then for 
any $s > \sqrt{2\pi/n} \cdot \lambda_1(\lat)$,
\begin{equation}
\label{eq:klboundongaussianbigs}
\rho_s(\lat) \le 1 + \left(\frac{\beta^2 s^2 n}{2  \pi   e \cdot \lambda_1(\lat)^2}\right)^{n/2+1}2^{o(n)} \;,
\end{equation}
and for $s \leq \sqrt{2\pi/n} \cdot \lambda_1(\lat)$, we have
\begin{align}
\rho_s(\lat) \le 1 + e^{-\pi \lambda_1(\lat)^2/s^2}\cdot \beta^{n+o(n)}
\; . \label{eq:klboundongaussian}
\end{align}
\end{lemma}

We note that an easy calculation shows that the right-hand side of Eq.~\eqref{eq:klboundongaussianbigs} is never smaller than the right-hand side of Eq.~\eqref{eq:klboundongaussian}. In particular, this means that Eq.~\eqref{eq:klboundongaussianbigs} actually applies for all $s$.

\begin{proof}[Proof of Lemma~\ref{lem:Gaussian-sum-bound}]
We assume without loss of generality that $\lat$ is normalized so that $\lambda_1(\lat) = 1$. Let $t := 1+1/n$.
For $r \ge 1$, define $T_r:= \{\vec{x} \in \R^n : r \le  \|\vec{x}\| < tr\}$.
By Lemma~\ref{lem:KL},
$
|\cL \cap T_r| \le \beta^{n+o(n)} r^n
$,
and, for any vector $\vec{y} \in \cL \cap T_r$, 
$
\rho_s(\vec{y}) \le e^{-\pi r^2/s^2}
$.
Therefore, 
\[
\rho_s(\cL \cap T_r) \le e^{-\pi r^2/s^2}\cdot \beta^{n + o(n)} \cdot r^n 
\; .
\]
So, we have
\begin{align*}
\rho_s(\lat) &= 1+ \sum_{i = 0}^\infty \rho_s(\cL \cap T_{t^i}) \\
&\leq 1 + \beta^{n + o(n)} \cdot \sum_{i = 0}^\infty  e^{-\pi t^{2i}/s^2} t^{in} \\
&\leq 1 + (1+s) \beta^{n + o(n)} \cdot \max_{r \geq 1}e^{-\pi r^2/s^2} r^n\; ,
\end{align*}
where we have used the fact that
$e^{-\pi t^{2i}/s^2} t^{in}$ decays geometrically when $i$ is at least, say, $(1+s)\cdot \poly(n)$, and so the sum up to that point is the same as the infinite sum up to a constant factor. 
Note that for any $a, b >0$, the maximum of 
$e^{-ar^2} r^b$ over the interval $r \ge 1$ is obtained at 
$r=\sqrt{b/(2a)}$ if this value is at least $1$ or at $r=1$ otherwise. The result follows.
\end{proof}

\begin{proposition}
\label{prop:boundonmass-KL}
Let $\lat \subset \R^n$ be a lattice of rank at least one. Let 
\[
s = \sqrt{\frac{2\pi e}{\beta^2 n}} \cdot \lambda_1(\lat)
\; .
\] 
Then,
 \[ \Pr_{\vec{X} \sim D_{\lat,s}}[\length{\vec{X}} = \lambda_1(\lat)] \geq e^{-\beta^2 n/(2e) - o(n)} \approx 1.38^{-n - o(n)}
\; .
\]
\end{proposition}
\begin{proof}
By Lemma~\ref{lem:Gaussian-sum-bound}, we have that $\rho_s(\lat) = 2^{o(n)}$. 
Therefore,
\[
\Pr_{\vec{X} \sim D_{\lat,s}}[\length{\vec{X}} = \lambda_1(\lat)]
\geq
e^{-\pi/s^2}/\rho_s(\lat)
\geq 
e^{-\pi/s^2 - o(n)} = e^{-\beta^2 n/(2e) - o(n)} 
\; ,
\]
as needed.
\end{proof}

An easy calculation shows that the probability in Proposition~\ref{prop:boundonmass-KL} is maximized to within a factor of two when $s = 1/\eta_1(\lat^*)$. I.e., for any shortest non-zero vector $\vec{y} \in \lat$, 
\[
\max_s \Pr_{ \vec{X} \sim D_{\lat, s}}[\vec{X} = \vec{y}] \leq 2 \Pr_{ \vec{X} \sim D_{\lat, 1/\eta_1(\lat^*)}}[\vec{X} = \vec{y}] = 
 \exp(-\pi \eta_1(\lat^*)^2 \cdot \lambda_1(\lat)^2)
\; .
\]
\subsection{A reduction from \texorpdfstring{$\problem{SVP}$}{SVP} to \texorpdfstring{$\problem{DGS}$}{DGS}}

\begin{theorem}
\label{thm:svptodgs}
There is a reduction from $\problem{SVP}$ to $\DGS{\frac{1}{2}}{}{2^{n/2}}$. The reduction makes $O(n)$ calls to the $\problem{DGS}$ oracle, preserves the dimension of the lattice, and runs in time $  2^{n/2} \cdot \poly(n)$.
\end{theorem}
\begin{proof}
Let $\alg{D}$ be an oracle solving $\DGS{\frac{1}{2}}{}{2^{n/2}}$. We construct an algorithm solving $\problem{SVP}$ as follows. It first runs the procedure from Theorem~\ref{thm:BKZ} on $\lat$ with $r=2$. Let $d$ be the length of the first basis vector in the output. 
For $i=0,\ldots, 100n $, the algorithm calls $\alg{D}$ on $\lat$ with parameter $s_i = 1.01^{-i} \cdot d$.  Let $\vec{x}_i$ be a shortest non-zero vector in the output. Finally, the algorithm outputs a shortest vector among the $\vec{x}_i$.

The running time of the algorithm is clear. By Theorem~\ref{thm:BKZ}, we have 
$d \leq 2^{n/2}\lambda_1(\lat)$. It follows that there exists some $i$ such that $\hat{s} \leq s_i \leq 1.01 \hat{s}$ where $\hat{s} = \sqrt{\pi e/n} \cdot 2^{0.099} \cdot \lambda_1(\lat)$ (i.e., $\hat{s}$ is the parameter from Proposition~\ref{prop:boundonmass-KL}). We assume that the output of $\alg{D}$ is exactly $D_{\lat, s_i}^{2^{n/2}}$ when called on $s_i$, incurring statistical distance at most $1/2$. For such $i$, by Lemma~\ref{lem:banaszczyk} we have 
\begin{align*}
\Pr_{\vec{X} \sim D_{\lat,s_i}}[\length{\vec{X}} = \lambda_1(\lat)] 
&\geq 1.01^{-n} \cdot \Pr_{\vec{X} \sim D_{\lat,\hat{s}}}[\length{\vec{X}} = \lambda_1(\lat)]\\
&\geq 1.4^{-n - o(n)} 
&\text{(Proposition~\ref{prop:boundonmass-KL})}
\; .
\end{align*}
The result follows by noting that $1.4 < \sqrt{2}$, so $2^{n/2}$ samples from $D_{\lat, s_i}$ will contain a shortest vector with probability at least $1-\exp(-\Omega(n))$.
\end{proof}

\begin{corollary}
\label{cor:SVP}
There is an algorithm that solves \problem{SVP} in time $2^{n + o(n)}$.
\end{corollary}
\begin{proof}
Combine the reduction from Theorem~\ref{thm:svptodgs} with the algorithm from Theorem~\ref{thm:DGS}.
\end{proof}

\section{Sampling \texorpdfstring{$2^{n/2}$}{2(n/2)} vectors above smoothing in \texorpdfstring{$2^{n/2}$}{2(n/2)} time}
\label{sec:abovesmooth}

In this section we present a $2^{n/2}$-time algorithm for $\problem{DGS}_\sigma$ with $\sigma$ approximately the smoothing parameter. For our applications (in particular for solving $O(1)\text{-}\problem{GapSVP}$) we will need a slightly stronger guarantee from the algorithm. Namely, when 
asked to produce samples with too small a parameter $s \le \sigma$, the 
output should still consist of discrete Gaussian samples of the desired parameter, but potentially less of them 
(or even none at all).
We make this property formal in the following slight modification of Definition~\ref{def:dgs}.

\begin{definition}
\label{def:hDGS}
For $\eps \geq 0$, $\sigma$ a function that maps lattices to non-negative real numbers, and $m \in \N$, $\hDGS{\eps}{\sigma}{m}$ (the \emph{honest} Discrete Gaussian Sampling problem) is defined as follows: 
The input is a basis $\basis$ for a lattice $\lat \subset \R^n$ and a parameter $s > 0$. The goal is for the output distribution to be $\eps$-close to $D_{\lat, s}^{m'}$ for some independent random variable $m' \geq 0$. If $s > \sigma(\lat)$, then $m'$ must equal $m$.
\end{definition}

\subsection{Sampling from the square root}

\begin{claim}
\label{clm:sqrt}
There is an algorithm that, given black-box access to the Bernoulli distribution $\bern(p)$ for unknown $0 < p \leq 1$, outputs $b \in \{0,1\}$ such that the distribution of $b$ is exactly $\bern(\sqrt{p})$. The expected running time of the algorithm is $O(1/p)$.
Furthermore, if the algorithm's input is restricted to $\tau \in \N$ independent samples from the Bernoulli distribution $\bern(p)$ for unknown $0 \leq p \leq 1$, then the distribution of $b$ is within statistical distance $(1-p)^\tau$ of $\bern(\sqrt{p})$, and the algorithm runs in time $O(\tau)$.
\end{claim}
\begin{proof}
The algorithm repeatedly samples from $\bern(p)$ until the first time that it sees a $1$. Let $k$ be the number of times that zero appears before this first one. (E.g., if the input sequence is $(0,0,0,1,\ldots)$, $k = 3$.) Then, the algorithm tosses $2k$ unbiased coins and outputs $1$ if exactly half of them are heads and 0 otherwise. (In the case when the algorithm's input length is restricted, it outputs $0$ if there are no ones in its input sequence.)

Correctness is immediate from the following identity, which can be obtained by taking the Taylor expansion of $p^{-1/2}$ around $p=1$ and then multiplying by $p$,
\[
\sqrt{p} = \sum_{k=0}^\infty \binom{2k}{k} 2^{-2k} (1-p)^k p \; .
\]
The expected running time is clearly proportional to $\sum (1-p)^k = 1/p$.
\end{proof}

\begin{definition}
For a vector $\vec{p} \in [0,1]^N$ with $\sum p_i = 1$, let $\sqrt{\vec{p}} = (\sqrt{p_1}/\psqrt, \ldots, \sqrt{p_N}/\psqrt)$ where $\psqrt = \sum \sqrt{p_i}$.
\end{definition}

\begin{theorem}[Square-root sampler]
\label{thm:sqrtsampler}
There is an algorithm that takes as input $\kappa \geq 2$ (the confidence parameter),
$t \geq 1$ (an upper bound on the ratio 
$\max_{i,j} p_i/p_j$), and $M$ elements from $\{1, \ldots, N \}$ and outputs a sequence of elements from the same set such that
\begin{enumerate}
\item \label{item:sqrtruntime} the running time is $M \cdot \poly(\log N, \kappa, t)$; 
\item \label{item:sqrtinputoutput} each $i \in \{1,\ldots, N\}$ appears at least as often in the input as in the output; and
\item  \label{item:sqrtdistribution} if the input consists of
$M \geq 4\kappa^4 t^2/\max p_i$
 independent samples from the distribution that assigns probability $p_i > 0$ to element $i$ with 
 $t \geq \max_{i,j} p_i/p_j$, then
the output is within statistical distance $C_1 M N \log N \exp(-C_2 \kappa)$ of $m$ independent samples with respective probabilities $\sqrt{\vec{p}}$ where $m = M/(16\kappa^3 t^{3/2})$.
\end{enumerate}
\end{theorem}
\begin{proof}
The algorithm first runs the procedure from Proposition~\ref{prop:estimatinglinfty} on the first $M/2$ elements from its input sequence, receiving as output $\approxpmax$. If $\approxpmax > 4t /N$, it fails. Otherwise, for $j=1,\ldots, M\approxpmax/3$, the algorithm samples
$r$ from $\pois(1/\approxpmax)$
and takes the next $r$ unused elements in the input (or fails if there are not $r$ unused elements remaining). For $i = 1,\ldots, N$, let $a_{i,j}$ be the number of times element $i$ appears in the $j$th such subsequence. For each $i,j$ let $b_{i,j}$ be $1$ with probability $\min\{ 1,\ a_{i,j}/\kappa \}$ and $0$ otherwise. Let $\tau = \ceil{\kappa^2 t}$. Then, for $i = 1,\ldots, N$ and $k = 0, \ldots, M \approxpmax/(3\tau)-1$, let $b_{i,k}^*$ be the output of the procedure from Claim~\ref{clm:sqrt} on input $(b_{i,\tau k+1},\ldots, b_{i,\tau (k+1)})$. 

Finally, the algorithm repeats the following at most $M/(5 \tau )$ times: 
it samples $i \in \{1,\ldots, N\}$ uniformly at random, and adds it to its output if $b_{i,k}^* = 1$ where $k \geq 1$ is the smallest index such that $b_{i,k}^*$ is unused (or it fails if there is no unused $b_{i,k}^*$). The algorithm stops as soon as its output contains $m$ samples. It fails if the output contains fewer than $m$ samples when the loop ends.

Note that the number of coins $b_{i,j}$ and $b_{i,k}^*$ is less than $MN\approxpmax \leq 4Mt$. It follows that the running time is as claimed. Note also that each $i$ appearing in the output corresponds to some $a_{i,j} > 0$, which in turn corresponds to at least one element $i$ in the input.

We first prove that the output elements have the correct distribution, ignoring failure.  By Proposition~\ref{prop:estimatinglinfty}, we can assume that $\max p_i \leq \approxpmax \leq 4\max p_i$, introducing statistical distance at most $C_1 N \log N \exp(-C_2 \kappa)$. By Lemma~\ref{lem:independence} and Claim~\ref{clm:poissontobinom}, we can further assume that $b_{i,j}$ are independent and distributed exactly as $\bern(p_i/(\kappa \approxpmax))$, introducing statistical distance that is at most $C_1 N M \exp(-C_2 \kappa)$. Applying Claim~\ref{clm:sqrt}, we have that $b_{i,k}^*$ is within statistical distance $(1-p_i/(\kappa\approxpmax))^{\tau} \leq \exp(-C \kappa)$ of $\bern(\sqrt{p_i/(\kappa \approxpmax)})$. So, the outputs have the correct distribution.

We now prove that the algorithm rarely fails. The number of inputs used in the second stage is distributed as $\pois(M/3)$, which by Corollary~\ref{cor:poistail} is at most $M/2$ except with probability at most $C_1 \exp(-C_2 \kappa)$. Applying the Chernoff-Hoeffding bound (Lemma~\ref{lem:chernoff}), we have that the number of coins $b_{i,k}^*$ used for a fixed $i$ is at most $M/(3\tau N) \leq M\approxpmax/(3\tau )$ except with probability at most $\exp(-C\kappa)$ (where we have used the fact that $M/(3\tau N) \geq \kappa$). Finally, applying the Chernoff-Hoeffding bound again, we have that after 
$M /(5\tau )$ steps, the size of the output will be at least 
$M/(8\tau ) \cdot \sqrt{\min p_i/(\kappa \approxpmax)} \geq m$
 except with probability at most $\exp(-C\kappa)$. So, the algorithm fails with probability at most $C_1 N \exp(-C_2 \kappa)$.
\end{proof}

In order to create an \scarequotes{honest} discrete Gaussian sampler as in Definition~\ref{def:hDGS} that uses the square-root sampler, we will need a way to \scarequotes{check $t$,} so that we only use the square-root sampler when we can be sure that Item~\ref{item:sqrtdistribution} applies. We achieve this with the following simple claim.

\begin{claim}
\label{clm:tcheck}
There is an algorithm that takes as input a number $t \geq 1$ and $M$ independent samples from the distribution that assigns probability $p_i > 0$ to each element $i \in \{1,\ldots, N\}$, runs in time $M\cdot \polylog(N, M, t)$, and satisfies 
\begin{enumerate}
\item \label{item:smalltcheck} if $t < \max p_i/p_j$, then the algorithm outputs no with probability at least $1-2\exp(-C M/(tN))$; and
\item \label{item:bigtcheck} if $t \geq 4\max p_i/p_j$, then the algorithm outputs yes with probability at least $1-N\exp(-C M/(tN))$.
\end{enumerate}
\end{claim}
\begin{proof}
The algorithm is quite simple. On input $X_1,\ldots, X_M$, let $\Tmax := \max_i |\{ j : X_j = i \} |$ and $\Tmin := \min_i |\{ j : X_j = i \} |$. The algorithm outputs no if $t\cdot \Tmin < 2\Tmax$ and yes otherwise.

The running time of the algorithm is clear. Let $\pmax = \max p_i$ and $\pmin = \min p_i$. Suppose $t < \pmax/\pmin$. Then, by the Chernoff-Hoeffding bound (Lemma~\ref{lem:chernoff}), we have $\Tmax > \pmax M /\sqrt{2}$ except with probability at most $\exp(-C\pmax M) \leq \exp(-C M/(tN))$. Similarly, we have that $\Tmin < \sqrt{2} \pmax M /t$ except with probability at most $\exp(-C \pmax M/t) \leq \exp(-C M/(tN))$. Item~\ref{item:smalltcheck} follows.

Now, suppose $t \geq 4\max p_i/p_j$. Then, for any $i$, by the Chernoff-Hoeffding bound we have $\pmin M/\sqrt{2} <|\{ j : X_j = i \} | < \sqrt{2} \pmax M$ except with probability at most $2\exp(-C M/(tN) )$. Item~\ref{item:bigtcheck} then follows by union bound.
\end{proof}

\subsection{A more efficient combiner that works above smoothing}

The following lemma generalizes the first part of Lemma~\ref{lem:sumofgaussians}. In particular, we recover Lemma~\ref{lem:sumofgaussians} when $\lat' = 2\lat$. 
(Note that, since we require that $2\lat \subseteq \lat'$, we have that the sum of two lattice vectors $\vec{X}_1 + \vec{X}_2$ is in $\lat'$ if and only if $\vec{X}_1$ and $\vec{X}_2$ are in the same coset of $\lat'$ over $\lat$.)
\begin{lemma}
\label{lem:anysublattice}
Let $\lat \subset \R^n$ be a lattice, and let $\lat' \subseteq \lat$ be a sublattice with $2\lat \subseteq \lat'$. Then for any $\vec{y} \in \lat'$ and $s > 0$, we have
\[
\Pr_{(\vec{X}_1, \vec{X}_2) \sim D_{\lat,s}^2}[\vec{X}_1 + \vec{X}_2 = \vec{y}\ |\ \vec{X}_1 + \vec{X}_2 \in \lat'] = \frac{\rho_{\sqrt{2} s}(2\lat + \vec{y})^2}{\pcol} \cdot \Pr_{\vec{X} \sim D_{2 \lat + \vec{y}, \sqrt{2} s}}[\vec{X} = \vec{y}]
\; ,
\]
where $\pcol = \sum_{\vec{d} \in \lat'/(2\lat)}  \rho_{\sqrt{2}s}(\vec{d})^2$.
\end{lemma}
\begin{proof}
It suffices to show that the probability on the left-hand side is proportional to $\rho_{\sqrt{2} s}(2\lat + \vec{y})\rho_{\sqrt{2} s}(\vec{y})$. Indeed,
\begin{align*}
\Pr_{(\vec{X}_1, \vec{X}_2) \sim D_{\lat, s}^2}[\vec{X}_1 + \vec{X}_2 = \vec{y}]  &= \frac{1}{\rho_s(\lat)^2}\cdot\sum_{\vec{x} \in \lat} \rho_s(\vec{x}) \rho_s(\vec{y} - \vec{x})\\
&= \frac{\rho_{\sqrt{2} s}(\vec{y})}{\rho_s(\lat)^2}\cdot\sum_{\vec{x} \in \lat} \rho_{\sqrt{2}s}(2\vec{x} - \vec{y})\\
&= \frac{\rho_{\sqrt{2} s}(\vec{y})}{\rho_s(\lat)^2}\cdot\rho_{\sqrt{2}s}(2\lat + \vec{y})
\; .
\end{align*}
\end{proof}

\begin{proposition}
\label{prop:sqrtcombiner}
There is an algorithm that takes as input a lattice $\lat \subset \R^n$, a sublattice $\lat' \subseteq \lat $ of index $2^a \geq 2^{n/2}$ with $2\lat \subseteq \lat'$, 
$\kappa \geq 2$ (the confidence parameter), 
and a sequence of vectors from $\lat$ such that, if the input consists of
$M \geq C \kappa^5 2^a$
independent samples from $D_{\lat, s}$ for some $s > 0$, then
\begin{enumerate}
\item the running time of the algorithm is $M \cdot \poly(n, \kappa)$;
\item \label{item:sqrtcombinerdistrib} the output distribution is 
$M \exp(C_1 n - C_2 \kappa)$-close to $m$ independent samples from $D_{\lat', \sqrt{2} s}$ 
where $m \in \{0, M/(C \kappa^4)\}$ is an independent random variable; and
\item \label{item:goodt} 
if $s \geq \eta_\eps(\lat')$ and $s \geq \sqrt{2}\eta_\eps(\lat)$, then $m = M/(C \kappa^4)$ where $\eps := 3/4$.
\end{enumerate}
\end{proposition}
\begin{proof} 
Let $(\vec{X}_1, \ldots, \vec{X}_M)$ be the input vectors,
and for each $i$, let $\coset_i \in \lat/\lat'$ be the coset of $\vec{X}_i$. The algorithm first applies the square sampler in a manner similar to that of the algorithm from Proposition~\ref{prop:combiner}.
Namely, the algorithm runs the procedure from Theorem~\ref{thm:squaresampler} with input $\kappa$
and $(\coset_1,\ldots,\coset_M)$,
receiving output $(\coset_1', \ldots, \coset_q')$.
For each $i=1,\ldots,q$, it chooses a pair of unpaired vectors $\vec{X}_j, \vec{X}_k$ with 
$\coset_j = \coset_k = \coset_i'$ and sets $\vec{Y}_i = \vec{X}_j + \vec{X}_k \in \lat'$. 

Let $t := (1+\eps)^2/(1-\eps)^2 = C$. The algorithm now applies two tests to \scarequotes{check that the distribution is sufficiently smooth.} First, it checks if 
$q \geq M/(32\kappa \sqrt{t})$.
If not, it halts and outputs nothing. Next, let $\vec{d}_i \in \lat'/(2\lat)$ be the coset of $\vec{Y}_i$. The algorithm runs the procedure from Claim~\ref{clm:tcheck} on the first $\floor{q/2}$ such cosets with parameter $t' := 4t$. It halts and outputs nothing if this procedure outputs no.

The algorithm now applies the square-root sampler to the remaining cosets.
Namely, it runs the procedure from Theorem~\ref{thm:sqrtsampler} with input $\kappa$, $t'$,
and $(\vec{d}_{\floor{q/2}+1},\ldots,\vec{d}_q)$, 
receiving output 
$(\vec{d}_1',\ldots, \vec{d}_L')$.
If $L < M/(C \kappa^4)$, it halts and outputs nothing.
Otherwise, for each $i \le M/(C \kappa^4)$, it chooses an unused vector $\vec{Y}_j$ with $j > q/2$ and
$\vec{d}_j = \vec{d}'_i$ and adds it to its output.

The running time of the algorithm follows from Item~\ref{item:squareruntime} of Theorem~\ref{thm:squaresampler} and the corresponding Item~\ref{item:sqrtruntime} of Theorem~\ref{thm:sqrtsampler}.
Furthermore, we note that by Item~\ref{item:squareinputoutput} of Theorem~\ref{thm:squaresampler}, the first step of the above algorithm will always be able to find unused $j,k$ satisfying $\coset_j = \coset_k = \coset_i'$, and by Item~\ref{item:sqrtinputoutput} of Theorem~\ref{thm:sqrtsampler}, the second step will always be able to find an unused $j$ satisfying $\vec{d}_j = \vec{d}'_i$.

We now prove Item~\ref{item:sqrtcombinerdistrib}. Note that, since the index of $\lat'$ over $\lat$ is $2^a $, the maximal probability of a coset must be at least $2^{-a}$. It follows that
$(\vec{c}_1,\ldots,\vec{c}_M)$
satisfy the conditions necessary for Item~\ref{item:squaredistribution} of Theorem~\ref{thm:squaresampler}.
Applying the theorem, we have that (up to statistical distance $M \exp(C_1 n - C_2 \kappa)$) the output vectors $(\vec{Y}_1,\ldots,\vec{Y}_q)$ are independent and
\begin{equation}
\label{eq:qbeforesmooth}
q \geq M \cdot \frac{1}{32 \kappa} \cdot \frac{\sum_{\vec{c} \in \lat/\lat'} \rho_s(\vec{c})^2}{\rho_s(\lat) \rho_s(\lat')}
\; .
\end{equation}
Furthermore, they assign to each $\vec{y} \in \lat'$ the probability
\begin{align*}
\Pr[\vec{Y}_i = \vec{y}] &= \frac{1}{\sum_{\coset \in \lat/\lat'} \rho_s(\coset)^2} \cdot \sum_{\coset \in \lat/\lat'}  \rho_s(\coset)^2 \cdot
\Pr_{(\vec{X}_1, \vec{X}_2) \sim D_{\coset, s}^2}[\vec{X}_1 + \vec{X}_2 = \vec{y}]\\
&=\Pr_{(\vec{X}_1, \vec{X}_2) \sim D_{\lat,s}^2}[\vec{X}_1 + \vec{X}_2 = \vec{y}\ |\ \vec{X}_1 + \vec{X}_2 \in \lat'] \\
&= \frac{\rho_{\sqrt{2} s}(2\lat + \vec{y})^2}{\sum_{\vec{d} \in \lat'/(2\lat)} \rho_{\sqrt{2} s}(\vec{d})^2} \cdot \Pr_{\vec{X} \sim D_{2 \lat + \vec{y}, \sqrt{2} s}}[\vec{X} = \vec{y}]
\; ,
\end{align*}
where we have used Lemma~\ref{lem:anysublattice}.
In particular, the distribution of each $\vec{d}_i$ is given by
\begin{align*}
\Pr[\vec{d}_i = \vec{d}] &= 
\frac{\rho_{\sqrt{2} s}(\vec{d})^2}{\sum_{\vec{d}' \in \lat'/(2\lat)} \rho_{\sqrt{2} s}(\vec{d'})^2} 
\; ,
\end{align*}
for $\vec{d} \in \lat'/(2\lat)$. The highest probability is obtained at $\vec{d}=2\lat$, and we denote it by $\pmax$.

Since the algorithm outputs nothing otherwise (in which case Item~\ref{item:sqrtcombinerdistrib} trivially holds), we only need to consider the case when 
\begin{equation}
\label{eq:boundonq}
q \geq \frac{M}{32\kappa \sqrt{t}}
\; ,
\end{equation} 
so we assume this below.  In addition, by Item~\ref{item:smalltcheck} of Claim~\ref{clm:tcheck}, the algorithm will halt after the second \scarequotes{smoothness test} with probability at least $1-2\exp(C\kappa)$ unless 
\begin{equation}
\label{eq:tprimebound}
t' \geq \frac{\rho_{\sqrt{2}s}(2\lat)^2}{\min_{\vec{d} \in \lat'/(2\lat)}\rho_{\sqrt{2}s}(\vec{d})^2}
\; .
\end{equation}
So, we can also assume that Eq.~\eqref{eq:tprimebound} holds.
Using Eq.~\eqref{eq:boundonq} and the fact that the index of $\lat'$ over $2\lat$ is $2^{n-a} \leq 2^a$, 
we see that $q/2 \geq 4\kappa^4 t^{\prime 2}/\pmax$.
Combining this with Eq.~\eqref{eq:tprimebound}, we see that the conditions for Item~\ref{item:sqrtdistribution} of Theorem~\ref{thm:sqrtsampler} are satisfied.
Let $\vec{W}_1,\ldots, \vec{W}_L$ be the vectors \scarequotes{chosen by the square-root sampler.} Applying the theorem, up to statistical distance $M \exp(C_1 n - C_2 \kappa)$, we have that the $\vec{W}_i$ are independently distributed, and 
\[
L =  \frac{q}{32 \kappa^3 t^{\prime 3/2}} 
\geq M \cdot \frac{1}{C \kappa^4}
\; ,
\] 
as needed, where we have used Eq.~\eqref{eq:boundonq}.
Furthermore, we have that for any coset $\vec{d} \in \lat'/(2\lat)$,
\[
\Pr[\vec{d}_i' = \vec{d}] = \frac{\rho_{\sqrt{2} s}(\vec{d})}{\sum_{\vec{d}' \in \lat'/(2\lat)} \rho_{\sqrt{2} s}(\vec{d}')} = \frac{\rho_{\sqrt{2} s}(\vec{d})}{\rho_{\sqrt{2} s}(\lat')}
\; .
\]  
Therefore, for any $\vec{y} \in \lat'$, we have
\[
\Pr[\vec{W}_i = \vec{y}] = \frac{\rho_{\sqrt{2} s}(2\lat + \vec{y})}{\rho_{\sqrt{2} s}(\lat')} \cdot \Pr_{\vec{X} \sim D_{2 \lat + \vec{y}, \sqrt{2} s}}[\vec{X} = \vec{y}]\ = \Pr_{\vec{X} \sim D_{\lat', \sqrt{2} s}}[\vec{X} = \vec{y}]
\; ,
\]
as needed.

Finally, we prove Item~\ref{item:goodt}. Suppose that $s$ satisfies $s \geq \eta_\eps(\lat')$ and $s \geq \sqrt{2}\eta_\eps(\lat)$. Note that
by Claim~\ref{clm:smooth}, we have that
\[
\frac{\rho_{s}(\lat')\rho_{s}(\lat)}{\sum_{\coset \in \lat/\lat'}\rho_{s}(\coset)^2} \leq \frac{1+\eps}{1-\eps} \cdot \frac{\rho_{s}(\lat)}{\sum_{\coset \in \lat/\lat'}\rho_{s}(\coset)} = \sqrt{t}
\; .
\] 
Combining this with Eq.~\eqref{eq:qbeforesmooth} shows that the algorithm will not halt after the first \scarequotes{smoothness test} except with probability at most $M \exp(C_1 n - C_2 \kappa)$. 
Similarly, since $\sqrt{2} s \geq \eta_\eps(2\lat)$,
\[
\frac{\rho_{\sqrt{2} s}(2\lat)^2}{\min_{\vec{d} \in \lat' /(2\lat)}\rho_{\sqrt{2} s}(\vec{d})^2}  
\leq t
\; .
\]
By applying Item~\ref{item:bigtcheck} of Claim~\ref{clm:tcheck}, we see that the algorithm also will not halt after the second \scarequotes{smoothness test} except with negligible probability. Therefore, Item~\ref{item:goodt} holds.
\end{proof}

We are going to apply Proposition~\ref{prop:sqrtcombiner} repeatedly, to a \scarequotes{tower} of lattices $(\lat_0,\ldots, \lat_\ell)$, as defined next. 

\begin{definition}
\label{def:tower}
For an integer $a$ satisfying $n/2 \leq a \leq n$,
we say that $(\lat_0,\ldots, \lat_\ell)$ is a \emph{tower of lattices in $\R^n$ of index $2^a$} 
if for all $i$ we have $2\lat_{i-1} \subseteq \lat_i \subset \lat_{i-1}$, $\lat_{i}/2 \subseteq \lat_{i-2}$, 
and the index of $\lat_i$ in $\lat_{i-1}$ is $2^a$.
\end{definition}

We next observe that it is easy to construct a tower with any desired final lattice $\lat_\ell$. In fact, one can even choose $\lat_{\ell-1}$, the second-to-last lattice in the tower.

\begin{claim}\label{clm:buildingtower}
There is a polynomial-time algorithm that given integers $\ell \ge 1$ and $n/2 \le a \le n$,
as well as two lattices $\lat$ and $\lat'$ in $\R^n$ satisfying 
$\lat \subseteq \lat' \subseteq \lat/2$ with the index of $\lat$ in $\lat'$ being $2^a$,
outputs a tower of lattices $(\lat_0,\ldots,\lat_\ell)$ of index $2^a$ with
$\lat_\ell = \lat$, $\lat_{\ell-1} = \lat'$, and 
$\lat_0 \supseteq 2^{-\floor{\ell a/n}} \lat$.
\end{claim}
\begin{proof}
Let $\vec{b}_1,\ldots,\vec{b}_n$ be a basis of $\lat$ chosen so that 
$\vec{b}_1/2,\ldots,\vec{b}_a/2,\vec{b}_{a+1},\ldots,\vec{b}_n$ is
a basis of $\lat'$. It is not difficult to see that such a basis exists. 
Then define the tower by ``cyclically halving $a$ coordinates,'' namely,
\begin{align*}
\lat_\ell &= \lat(\vec{b}_1,\ldots,\vec{b}_n), \\
\lat_{\ell-1} &= \lat(\vec{b}_1/2, \ldots, \vec{b}_{a}/2, \vec{b}_{a + 1}, \ldots, \vec{b}_{n}), \\
\lat_{\ell-2} &=  \lat(\vec{b}_1/4, \ldots, \vec{b}_{2a - n}/4, \vec{b}_{2a -n + 1}/2, \ldots, \vec{b}_{n}/2),
\end{align*}
etc. It is easy to check that this satisfies all the required properties.
\end{proof}

\begin{corollary}
\label{cor:sqrtpipeline}
There is an algorithm that takes as input a tower of lattices $(\lat_0, \ldots, \lat_\ell)$ in $\R^n$ of index $2^a \geq 2^{n/2}$, 
$\kappa \geq 2$ (the confidence parameter), and 
$
M =  (C \kappa^4)^{\ell+1} \cdot 2^a
$
vectors in $\lat_0$ such that, 
\begin{enumerate}
\item \label{item:towerruntime} the algorithm runs in time $M \cdot \poly(n, \kappa, \ell)$;
\item \label{item:anyparameter} if the input vectors are distributed as $D_{\lat_0, s}$ for some $s$, then the output is 
$M \ell  \exp(C_1 n-C_2\kappa )$-close to $m$ independent samples from $D_{\lat_\ell, 2^{\ell/2} s}$ where $m \in \{0,2^{n/2}\}$ is an independent random variable; and
\item \label{item:sabovesmooth}
if $2^{\ell/2} s \geq 2 \eta_{3/4}(\lat_{\ell-1})$ and $2^{\ell/2} s \geq \sqrt{2}\eta_{3/4}(\lat_\ell)$,
then $m = 2^{n/2}$.
\end{enumerate}
\end{corollary}
\begin{proof}
Let  $\mathcal{X}_0 = (\vec{X}_1,\ldots, \vec{X}_M)$ be the sequence of input vectors.
For $i = 0,\ldots, \ell-1$, the algorithm calls the procedure from Proposition~\ref{prop:sqrtcombiner} with input $\lat_{i}$, $\lat_{i+1}$, 
$\kappa$, 
and $\mathcal{X}_{i}$, receiving output $\mathcal{X}_{i+1}$. 
If $\mathcal{X}_{i+1}$ is empty, it halts and outputs nothing. Finally, the algorithm outputs the first $2^{n/2}$ vectors in $\mathcal{X}_\ell$. 

The running time is clear. 
Define $M_i = M/(C\kappa^4)^i$. 
Since $M_i \ge C \kappa^5 2^a$ for $0 \le i \le \ell-1$, we have by 
induction using Item~\ref{item:sqrtcombinerdistrib} of Proposition~\ref{prop:sqrtcombiner}
that for $i=0,\ldots,\ell$, up to statistical distance $i M\exp(C_1 n-C_2\kappa )$,
$\mathcal{X}_i$ is distributed like $m_i$ independent random samples from
$D_{\lat_i,2^{i/2}s}$
where $m_i \in \{0,M_i\}$ is an independent random variable. 
Here for convenience, if the algorithm aborts at some stage $j$, we define $\mathcal{X}_i$ for $i>j$ as the empty set.
Since $M_\ell > 2^{n/2}$, Item~\ref{item:anyparameter} follows. 

Finally, suppose $2^{\ell/2} s \geq \max \{ 2 \eta_{3/4}(\lat_{\ell-1}),\ \sqrt{2} \eta_{3/4}(\lat_\ell) \}$. Since $\lat_{i}/2 \subseteq \lat_{i-2}$, we have that $\eta_{3/4}(\lat_{i-2}) \leq \eta_{3/4}(\lat_i)/2$. It follows that $2^{i/2} s \ge \max\{\sqrt{2}\eta_{3/4}(\lat_i),\ \eta_{3/4}(\lat_{i+1})\}$ for all $i=0,\ldots,\ell-1$.
Item~\ref{item:sabovesmooth} then follows immediately from Item~\ref{item:goodt} of Proposition~\ref{prop:sqrtcombiner}.
\end{proof}

\subsection{Sampling above smoothing in time \texorpdfstring{$2^{n/2}$}{2(n/2)}}

\begin{theorem}
\label{thm:smoothDGS}
Let $\sigma $ be the function that maps a lattice $\lat$ to 
$\sqrt{2} \cdot \eta_{1/2}(\lat)$. Then, there is an algorithm 
that solves $\hDGS{\exp(-\Omega(\kappa))}{\sigma}{2^{n/2}}$ in 
time $2^{n/2 + \polylog(\kappa) + o(n)}$ for any 
$\kappa \geq \Omega(n)$.
\end{theorem}
\begin{proof}
We first present an algorithm that works for $\sigma(\lat) =  2 \eta_{3/4}(\lat)$ and then modify it to achieve
the desired $\sigma(\lat) = \sqrt{2} \cdot \eta_{1/2}(\lat)$. 
On input $\lat \subset \R^n$ a lattice of rank $n$ and $s > 0$, the algorithm behaves as follows. 
It first
applies the algorithm from Claim~\ref{clm:buildingtower}
with parameters $a > n/2$ and $\ell \geq 1$ to be set in the analysis,
the lattice $\lat$, and an arbitrary choice of $\lat'$ satisfying the properties there.
It obtained a tower of lattices $(\lat_0, \ldots, \lat_\ell)$ of 
index $2^a$ such that $\lat_\ell = \lat$ and 
$\lat_0 \supseteq 2^{-\floor{\ell a/n}} \lat$.

The algorithm then runs the sampler from Proposition~\ref{prop:startgauss} on $\lat_0$ with parameters $r$ (to be set in the analysis), $\hat{s} = 2^{-\ell/2} s$, and 
 $M = (C\kappa^4)^{\ell + 1} 2^a$. 
 It receives as output a sublattice $\lat'_0 \subseteq \lat_0$ and vectors $\vec{X}_1,\ldots, \vec{X}_M \in \lat'_0$. 
If $\lat'_0 \neq \lat_0$, it outputs nothing and halts.
Otherwise, it runs the procedure from Corollary~\ref{cor:sqrtpipeline} with input $(\lat_0,\ldots, \lat_\ell)$, $\kappa$, and $(\vec{X}_1,\ldots, \vec{X}_M)$ and outputs the result.

Let $a = \ceil{n/2 + Cn/\log n}$, $\ell = C\ceil{\log^4 n}$, and $r = Cn/\log n$.
Applying Proposition~\ref{prop:startgauss} and Item~\ref{item:anyparameter} 
of Corollary~\ref{cor:sqrtpipeline}, we have that the output will be 
distributed as $D_{\lat,s}^m$ for some $m \in \{0,2^{n/2}\}$ up to statistical distance 
$\exp(-\Omega(\kappa))$, as needed. 
We wish to show that, 
if $s > 2 \eta_{3/4}(\lat)$, then we have $m = 2^{n/2}$. 
 Note that 
\[ \hat{s} > 2^{-\ell/2 + 1} \eta_{3/4}(\lat) \geq 2^{\ell(a/n - 1/2)} \eta_{3/4}(\lat_0) \geq (Cr)^{n/r} \sqrt{n\log n} \cdot \eta_{3/4}(\lat_0)
\; .
\]
Therefore, by Proposition~\ref{prop:startgauss}, we have that $\lat'_0 = \lat_0$, so that the algorithm will not halt after running the sampler from Proposition~\ref{prop:startgauss}. Furthermore, since $s >  2 \eta_{3/4}(\lat) = 2 \eta_{3/4}(\lat_\ell)$, we have
$2^{\ell/2} \hat{s} > \sqrt{2} \eta_{3/4}(\lat_\ell)$, and since $\lat_{\ell-1} \supset \lat_\ell$, we obviously also have 
\begin{align}\label{eq:conditiononhats}
2^{\ell/2} \hat{s} > 2 \eta_{3/4}(\lat_{\ell-1})\;.
\end{align}
Therefore, by Item~\ref{item:sabovesmooth} of Corollary~\ref{cor:sqrtpipeline}, we have that $m = 2^{n/2}$ as needed.

Now, consider the running time. The tower of lattices can be built in polynomial time. The procedure from Proposition~\ref{prop:startgauss} runs in time $(2^{O(r)} + M) \cdot \poly(n)$, and the procedure from Corollary~\ref{cor:sqrtpipeline} runs in time $M \cdot \poly(n, \kappa, \ell)$. It follows that the running time is as claimed.

We now show how to modify the above algorithm to work for $\sigma(\lat) =  \sqrt{2} \cdot \eta_{1/2}(\lat)$. The bottleneck in the above proof is the condition in 
Eq.~\eqref{eq:conditiononhats} needed for Item~\ref{item:sabovesmooth} of 
Corollary~\ref{cor:sqrtpipeline} to apply.
The trouble is that we used the
trivial inequality $\eta_{3/4}(\lat_{\ell-1}) \leq \eta_{3/4}(\lat_\ell)$ in order to show that this holds,
even though $\lat_{\ell-1}$ is a superlattice of  $\lat_\ell$ of index 
greater than $2^{n/2}$, and so one might expect a gap of about $\sqrt{2}$ between
these two quantities. Indeed, Lemma~\ref{lem:supersmoothlattice} below shows how to randomly choose such a superlattice $\lat_{\ell-1}$ such that $\eta_{3/4}(\lat_{\ell-1}) \leq \eta_{1/2}(\lat_\ell)/\sqrt{2}$ holds with constant positive probability. So we now use the same procedure as above, except we apply the algorithm 
in Claim~\ref{clm:buildingtower} with that choice of $\lat_{\ell-1}$.
Assuming $\lat_{\ell-1}$ satisfies this constraint, the constraint~\eqref{eq:conditiononhats} holds,
whenever $s > \sqrt{2} \eta_{1/2}(\lat)$ and hence the algorithm would be successful.
This almost completes the proof, except for one minor caveat: 
as described above, 
our algorithm successfully outputs $2^{n/2}$ vectors (in the ``good'' case of 
$s > \sqrt{2}\eta_{1/2}(\lat)$) only with some constant positive probability,
whereas our goal is to be successful with probability $1-\exp(-\kappa)$. 
This can easily be mended by repeating the algorithm $\kappa$ times, 
each time choosing an independent $\lat_{\ell-1}$. 
\end{proof}

\begin{lemma}
\label{lem:supersmoothlattice}
There is a probabilistic polynomial-time algorithm that takes as input a lattice $\lat \subset \R^n$ of rank $n$
and an integer $a$ with $n/2 \leq a < n$
and returns a superlattice $\lat' \supset \lat$ of index $2^a$ with $\lat' \subseteq \lat/2$ such that for any $\eps \in (0,1)$, we have $\eta_{\eps'}(\lat') \leq \eta_\eps(\lat)/\sqrt{2}$ with probability at least $1/2$, where $\eps' := 2\eps^2 + 2^{n/2+1-a}(1+\eps)$.
\end{lemma}
\begin{proof}
The algorithm simply selects a superlattice $\lat' \supset \lat$ of index $2^a$ with $\lat' \subseteq \lat/2$ uniformly at random. 
It will be convenient to equivalently work in the dual and to instead pick $\lat^{\prime *} \subset \lat^*$ of index $2^a$ with $2\lat^* \subseteq \lat^{\prime *}$. 
In more detail, let $(\vec{b}_1^*, \ldots, \vec{b}_n^*)$ be a basis of the dual lattice $\lat^*$. This defines
a group isomorphism $h:\F_2^n \to \lat^*/(2\lat^*)$ given by $h(\vec{a}) = \sum_i a_i \vec{b}_i^* \bmod 2\lat^*$. 
The algorithm picks a random subspace $V \subseteq \F_2^n$ of dimension $n-a$ and 
sets $\lat^{\prime *} = h(V)$ to be the union of the cosets corresponding 
to the points in $V$. (It can do this efficiently by, e.g., taking a basis 
$\vec{v}_1,\ldots,\vec{v}_{n-a}$ of $V$, and taking the lattice generated 
by $(2\vec{b}_1^*,\ldots, 2\vec{b}_n^*,\vec{y}_1,\ldots,\vec{y}_{n-a})$ where $\vec{y}_i$ is any coset representative of $h(\vec{v}_i)$.) It then returns the primal lattice $\lat'$.

It is clear that the algorithm runs in polynomial time and that $\lat$ has index $2^a$ over $\lat'$ with $\lat \subset \lat' \subseteq \lat/2$ as needed. Note that all vectors in $\F_2^n \setminus \{ \vec0\}$ have equal probability $(2^{n-a}-1)/(2^n-1)$ of being in the subspace $V$. Therefore, for any dual coset $\coset^* \in \lat^*/(2\lat^*)$ with $\coset^* \neq 2\lat^*$, we have
$
\Pr[\vec{c}^* \in \lat^{\prime *}] = (2^{n-a}-1)/(2^n-1)$. Then, assuming without loss of generality that $\eta_\eps(\lat) = 1$, we have 
\begin{align*}
\expect\big[\rho_{\sqrt{2}}(\lat^{\prime *})\big] &= \sum_{\coset^* \in \lat^*/(2\lat^*)} \Pr[\coset^* \in \lat^{\prime *}] \rho_{\sqrt{2}}(\coset^*)\\
&= \rho_{\sqrt{2}}(2\lat^*) + \frac{2^{n-a}-1}{2^n-1} \cdot  \sum_{\coset^* \in \lat^*/(2\lat^*) \setminus \{ 2\lat^*\}} \rho_{\sqrt{2}}(\coset^*)\\
&< 1+ \eps^2 + 2^{-a}\rho_{\sqrt{2}}(\lat^*) &\text{(Lemma~\ref{lem:doublesmooth})}\\
&\leq 1 + \eps^2 + 2^{n/2-a}(1+\eps) &\text{(Lemma~\ref{lem:banaszczyk})}
\; .
\end{align*}
By Markov's inequality, $\rho_{\sqrt{2}}(\lat^{\prime *} \setminus \{\vec0 \}) < 2\eps^2 + 2^{n/2+1-a}(1+\eps)$  with probability at least $1/2$, and 
the result follows.
\end{proof}

\subsection{Sampling from shifted Gaussians above smoothing}
\label{sec:shifted}

Here we observe that the sampler described above can actually be used to obtain $2^{n/2}$ samples from the \emph{shifted} discrete Gaussian $D_{\lat- \vec{t}, s}$ in $2^{n/2 + o(n)}$ time for any parameter $s > \sqrt{2} \cdot \eta_{\eps}(\lat)$ with $\eps \approx 1/2$. 
(Section~\ref{sec:approx-cvp} describes essentially the same reduction in a slightly different context.)
We present a brief proof sketch here in case this finds applications in future work. The idea is to call the sampler from Theorem~\ref{thm:smoothDGS} repeatedly on the lattice $\overline{\lat} := \lat(\vec{b}_1,\ldots, \vec{b}_n, \bar{\vec{t}}) \subset \R^{n+1}$, where $\bar{\vec{t}} := (-\vec{t}, s) \in \R^{n+1}$. Note that the lattice hyperplane $\lat + \bar{\vec{t}} \subset \overline{\lat}$ is simply a copy of $\lat - \vec{t}$ shifted by $s \vec{e}_{n+1}$, so that $\pi_{\R^n}(D_{\lat + \bar{\vec{t}}}) = D_{\lat - \vec{t}, s}$. We therefore simply return the first $n$ coordinates of the first $2^{n/2}$ vectors in $\lat + \bar{\vec{t}}$ that occur in the output. 

To prove that this algorithm works, we simply need to show that (1) $\eta_{\eps}(\lat) > \eta_{1/2}(\overline{\lat})$, so that the call to the algorithm from Theorem~\ref{thm:smoothDGS} will be valid as long as $s > \sqrt{2} \cdot \eta_{\eps}(\lat)$; and (2) when $s$ is above smoothing, a vector sampled from $D_{\overline{\lat}, s}$ will land in $\lat + \bar{\vec{t}}$ with relatively high probability, so that we will not have to make too many calls to the algorithm from Theorem~\ref{thm:smoothDGS} in order to find $2^{n/2}$ vectors in $\lat + \bar{\vec{t}}$. Both claims follow from standard calculations. (As described above, the algorithm achieves $\eps \approx 0.38$ and makes a constant number of calls to the centered DGS oracle. If we instead set $\bar{\vec{t}} := (-\vec{t}, s/\kappa)$ for $\kappa \geq 1$ and make $O(\kappa)$ oracle calls, we can obtain $\eps \approx 1/2 - e^{-C\kappa^2}$.)
\section{Solving \texorpdfstring{$O(1)$-\problem{GapSVP}}{O(1)-GapSVP} in \texorpdfstring{$2^{n/2 + o(n)}$}{2(n/2)+o(n)} time}
\label{sec:gapSVP}

In this section we present our $\problem{GapSVP}$ algorithm. The main idea is to approximate the smoothing
parameter of $\lat^*$ and then use Lemma~\ref{lem:etalambda1} to relate it to $\lambda_1(\lat)$. 
To distinguish a parameter above smoothing from a parameter below smoothing, we call the $\problem{hDGS}$ oracle with the given parameter. It is below smoothing if the oracle does not produce enough samples or if a statistical test on the output (Lemma~\ref{lem:smoothcovariance}) fails.

\begin{lemma}
\label{lem:etalambda1}
For any lattice $\lat \subset \R^n$ and $\eps \in (0,1)$, if $\eps > (e/\beta^2 + o(1))^{-n/2}$, we have
\begin{equation}
\label{eq:etalambda1big}
\sqrt{\frac{\log(1/\eps)}{\pi}}  < \lambda_1(\lat) \eta_\eps(\lat^*) < \sqrt{\frac{\beta^2 n}{2\pi e}} \cdot \eps^{-1/n} \cdot (1+o(1))
\;, 
\end{equation}
and if $\eps \leq (e/\beta^2 + o(1))^{-n/2}$, we have
\begin{equation}
\label{eq:etalambda1small}
\sqrt{\frac{\log(1/\eps)}{\pi}}  < \lambda_1(\lat) \eta_\eps(\lat^*) < 
\sqrt{\frac{\log(1/\eps) +  n \log \beta + o(n)}{\pi}} 
\; ,
\end{equation}
where $\beta := 2^{0.401}$.
\end{lemma}

As will be apparent from the proof and the remark after Lemma~\ref{lem:Gaussian-sum-bound}, Eq.~\eqref{eq:etalambda1big} actually holds for all $\eps \in (0,1)$.

\begin{proof}[Proof of Lemma~\ref{lem:etalambda1}]
Throughout the proof we assume without loss of generality that $\lambda_1(\lat) = 1$. 
For the lower bound in both cases, let $s \leq \sqrt{\log(1/\eps)/\pi}$. Then,
$
\rho_{1/s}(\lat \setminus \set{\vec{0}}) > e^{-\pi s^2} \geq \eps
$,
as needed. 

Let $\eps > (e/\beta^2 + o(1))^{-n/2}$, and let
$s$ be the expression in the right-hand side of~\eqref{eq:etalambda1big}.
Then, noting that $s < \sqrt{n/(2\pi)}$,
by Eq.~\eqref{eq:klboundongaussianbigs} in Lemma~\ref{lem:Gaussian-sum-bound}, we have 
\[
\rho_{1/s}(\lat \setminus \{ \vec0 \}) \leq \left( \frac{\beta^2 n}{2\pi e s^2} \right)^{n/2+1} 2^{o(n)} < \eps
\; ,
\]
as needed. 

Now, let $\eps \leq (e/\beta^2 + o(1))^{-n/2}$, and let 
$s$ be the expression in the right-hand side of~\eqref{eq:etalambda1small}.
Then, noting that $s \geq \sqrt{n/(2\pi)}$,
by Eq.~\eqref{eq:klboundongaussian} in Lemma~\ref{lem:Gaussian-sum-bound}, we have 
\begin{align*}
\rho_{1/s}(\lat \setminus \{ \vec0 \}) &< e^{-\pi s^2}\cdot \beta^{n+o(n)} \leq \eps \; ,
\end{align*}
as needed.
\end{proof}

\begin{definition}
For a matrix $M \in \R^{n\times n}$, the spectral norm of $M$ is defined as 
\[ \length{M} := \sup_{\length{\vec{x}} = 1}\length{M\vec{x}}
\; .
\]
\end{definition}
For a symmetric matrix $M$ (the only case that interests us), $\length{M}$ is equivalently the largest absolute value of an eigenvalue of $M$.

\begin{lemma}
\label{lem:smoothcovariance}
For any lattice $\lat \subset \R^n$ with $n\geq 3$ and any $\eps>0$,
\[
 \frac{\max \{ 1,\, \log(1/\eps) \}}{\pi n} \cdot \frac{\eps}{1+\eps}
< \Big\| \frac{1}{2\pi} \cdot I_n-\frac{1}{\eta_\eps(\lat)^2}\cdot  \expect_{\vec{X} \sim D_{\lat, \eta_\eps(\lat)}}[\vec{X}\vec{X}^T]\Big\|  
\leq \frac{1}{\pi} \cdot  \frac{\eps}{1+\eps} \cdot \Big(\log \frac{2(1+\eps)}{\eps} +1 \Big)
\; ,
\]
where $I_n$ is the $n\times n$ identity matrix.
\end{lemma}
\begin{proof}
For the upper bound, see {\cite[Lemma 4.4]{cvpp}}.

For the lower bound, from the same source, we have
\[
\frac{1}{2\pi} \cdot I_n - \frac{1}{\eta_\eps(\lat)^2} \cdot \expect_{\vec{X} \sim D_{\lat, \eta_\eps(\lat)}}[\vec{X}\vec{X}^T] = \eta_\eps(\lat)^2 \cdot \expect_{\vec{Y} \sim D_{\lat^*, 1/\eta_\eps(\lat)}}[\vec{Y}\vec{Y}^T]
\; .
\]
Note that for any positive semidefinite matrix $A \in \R^{n\times n}$, we have $\length{A} \geq \Tr(A)/n$. Therefore, 
\begin{align*}
\eta_\eps(\lat)^2 \cdot \Big\| \expect_{\vec{Y} \sim D_{\lat^*, 1/\eta_\eps(\lat)}}[\vec{Y}\vec{Y}^T] \Big\| &\geq \frac{\eta_\eps(\lat)^2}{n} \cdot \Tr\Big(\expect_{\vec{Y} \sim D_{\lat^*, 1/\eta_\eps(\lat)}}[\vec{Y}\vec{Y}^T] \Big) \\
&\geq \frac{\eta_\eps(\lat)^2}{n} \cdot \frac{\eps \lambda_1(\lat^*)^2}{1+\eps}
\; .
\end{align*}
Applying Lemma~\ref{lem:etalambda1} gives us the desired lower bound when $\eps \leq 1/e$. For $\eps > 1/e$, we obtain the result whenever $\eta_\eps(\lat) \cdot \lambda_1(\lat^*) > 1/\sqrt{\pi}$, 
so it suffices to consider the case $\eta_\eps(\lat)\lambda_1(\lat^*) \leq 1/\sqrt{\pi}$. Let $\vec{v} \in \lat^*$ satisfy $\length{\vec{v}} = \lambda_1(\lat^*)$. Then, applying Lemma~\ref{lem:subgaussianity}, we have
\begin{align*}
\expect_{\vec{X} \sim D_{\lat, \eta_\eps(\lat)} }\big[\inner{\vec{X},\vec{v}}^2/\lambda_1(\lat^*)^2\big] 
&= 2\int_0^\infty r \cdot \Pr[\abs{\inner{\vec{X}, \vec{v}}} \geq r\lambda_1(\lat^*)] {\rm d} r \\
&= \frac{\Pr[\abs{\inner{\vec{X}, \vec{v}}} \geq 1]}{\lambda_1(\lat^*)^2}+2\int_{1/\lambda_1(\lat^*)}^\infty r \Pr[\abs{\inner{\vec{X}, \vec{v}}} \geq r\lambda_1(\lat^*)] {\rm d} r \\
&\leq \frac{2e^{-\pi/(\lambda_1(\lat^*)\eta_\eps(\lat))^2}}{\lambda_1(\lat^*)^2} + 4\int_{1/\lambda_1(\lat^*)}^\infty r e^{-\pi r^2/ \eta_\eps(\lat)^2} {\rm d} r  \\
&= 2e^{-\pi/(\lambda_1(\lat^*)\eta_\eps(\lat))^2} \cdot \Big( \frac{1}{\lambda_1(\lat^*)^2} + \frac{\eta_\eps(\lat)^2}{\pi} \Big)\\
&\leq 2\frac{\pi + 1/\pi}{e^{\pi^2}} \cdot \eta_\eps(\lat)^2
\; ,
\end{align*}
where we have used the fact that $e^{-\pi/(\lambda \eta)^2}/\lambda^2$ is increasing for $\lambda \in (0,\sqrt{\pi}/\eta)$.
Therefore, 
\[
\Big\| \frac{1}{2\pi} \cdot I_n - \frac{1}{\eta_\eps(\lat)^2}\cdot  \expect_{\vec{X} \sim D_{\lat, \eta_{\eps}(\lat)}}[\vec{X}\vec{X}^T]\Big\| \geq \frac{1}{2\pi} - 2\frac{\pi + 1/\pi}{e^{\pi^2}} > \frac{1}{\pi n} \cdot \frac{\eps}{1+\eps}
\; ,
\]
as needed.
\end{proof}

We will also need a form of the matrix Chernoff bound.  In particular, we use a less general version of~\cite[Theorem 5.29]{Vershynin2012}. 

\begin{lemma}[Matrix Chernoff bound]
\label{lem:matrixchernoff}
Let $A_1, \ldots, A_N$ be independent and identically distributed random symmetric matrices in $\R^{n\times n}$ with $\length{A_i} \leq a$ and expectation $\mu$. Then, for any $t \in (0,a)$,
\[
\Pr\Big[ \Big\| \frac{1}{N}\sum A_i - \mu \Big\| \geq t \Big] \leq 2n \exp(-CN t^2/a^2)
\; .
\]
\end{lemma}

\begin{theorem}
\label{thm:gapSVP}
For any $\eps \in [2^{-n/2}, 1/e]$, 
there is a reduction from $\gamma\text{-}\problem{GapSVP}$ to $\hDGS{\frac{1}{4}}{\sqrt{2} \eta_{1/2}}{m}$ where $m := n^5/\eps^2$ and
\[
\gamma := \sqrt{\frac{ \beta^2 n + o(n)}{e \log(1/\eps)}} \cdot \Big(1+ \frac{2 \log n}{ \log(1/\eps)} \Big)
\; ,
\]
where $\beta := 2^{0.401}$.
The reduction preserves dimension, makes a single call to the $\problem{hDGS}$ oracle, and runs in time $m \cdot \poly(n)$.
\end{theorem}
\begin{proof}
On input a lattice $\lat \subset \R^n$ and $d > 0$, the reduction calls the $\problem{hDGS}$ oracle with input $\lat^*$ and parameter $s>0$ to be set in the analysis. If the oracle outputs fewer than $m$ vectors, the reduction immediately outputs yes (i.e., the reduction guesses that $\lambda_1(\lat) < d$). Otherwise, it receives as output $\vec{X}_1,\ldots, \vec{X}_{m} \in \lat^*$. Let $\Sigma := \frac{1}{m} \cdot \sum \vec{X}_i\vec{X}_i^T$ be the sample covariance. If

\[
\Big\|  \frac{1}{2\pi} \cdot I_n - \frac{1}{s^2} \cdot \Sigma\Big\| < \frac{\eps}{10 n} \cdot \log(1/\eps)
\; ,
\]
the reduction outputs no (i.e., it guesses that $\lambda_1(\lat) \geq \gamma \cdot d$). Otherwise, it outputs yes.

The running time is clear. Let
\[
s := \sqrt{\frac{\log (1/\eps)}{\pi}} \cdot \frac{1}{d} \; .
\]
Suppose $\lambda_1(\lat) < d$. Then, by the lower bound in Lemma~\ref{lem:etalambda1}, we have $s < \eta_\eps(\lat^*)$.
By the definition of $\problem{hDGS}$, we have that the output of the oracle is statistically close to $D_{\lat^*, s}^{m'}$ for some independent random variable $0 \leq m' \leq m$. So, we assume that the oracle outputs exactly this distribution, introducing statistical distance at most 1/4. If $m' < m$, then the reduction correctly outputs yes. Conditioning on $m' = m$ and using Lemma~\ref{lem:smoothcovariance}, we have 
\[
\Big\| \frac{1}{2\pi} \cdot I_n-\frac{1}{s^2} \cdot \expect[\Sigma] \Big\| > \frac{\eps}{2 \pi n}\cdot \log(1/\eps) 
\; ,
\]
where we have used the fact that the lower bound in Lemma~\ref{lem:smoothcovariance} is monotonically increasing.
So, in order to show that the reduction will output yes, it suffices to show that $\Sigma$ is concentrated around its mean. By Lemma~\ref{lem:banaszczyk} and union bound, we can assume that $\length{\vec{X}_i} \leq 100\sqrt{n}s$, introducing only negligible statistical distance. Assuming that this is the case, we can apply Lemma~\ref{lem:matrixchernoff} with $a = 100^2n$ and $t = \eps/(1000n)$, and we have that 
\[
\Pr\Big[ \Big\| \frac{1}{s^2} \Sigma - \frac{1}{s^2}  \expect[\Sigma]\Big\| \geq t \Big] \leq 2n\exp(-Cm t^2/a^2) \leq \exp(-Cn)
\; ,
\] 
where we have used the fact that $mt^2/a^2 \geq Cn$. It follows that the reduction correctly outputs yes with all but negligible probability.

Now, suppose $\lambda_1(\lat) \geq \gamma \cdot d$. Then, by Eq.~\eqref{eq:etalambda1big} of Lemma~\ref{lem:etalambda1} with $\eps$ there taken to be $1/2$, we have $s > \sqrt{2} \eta_{1/2}(\lat^*)$. In this regime, by the definition of $\problem{hDGS}$, we have that the output of the oracle is within statistical distance $1/4$ of $D_{\lat^*, s}^{m}$. So, we can assume that the output is exactly $D_{\lat^*, s}^{m}$, introducing statistical distance at most $1/4$. Applying Lemma~\ref{lem:etalambda1} again, we have 
\[ s > (1+2 \log n/\log(1/\eps))\eta_{\eps}(\lat^*)
> \eta_{\eps/n^2}(\lat^*)
\; ,
\] where we have used Lemma~\ref{lem:doublesmooth}, the fact that $\eps > 2^{-n/2}$, and the observation that Eq.~\eqref{eq:etalambda1big} applies for all $\eps \in (0,1)$. Applying Lemma~\ref{lem:smoothcovariance}, we have that 
\[
\Big\| \frac{1}{2\pi} \cdot I_n-\frac{1}{s^2} \cdot \expect[\Sigma] \Big\| < \frac{\eps}{\pi n^2} \cdot \Big( \log \frac{2(n^2+\eps)}{\eps} + 1\Big) < \frac{\eps}{20 n} \cdot \log(1/\eps)
\; ,
\]
for sufficiently large $n$ (where we have used the fact that the upper bound in Lemma~\ref{lem:smoothcovariance} is monotonically increasing).
Finally, applying Lemma~\ref{lem:matrixchernoff} as above shows that the oracle correctly outputs no with all but negligible probability.
\end{proof}

\begin{corollary}
\label{cor:gapSVP}
There is a randomized algorithm that solves $\gamma\text{-}\problem{GapSVP}$ for $\gamma := 1.93 + o(1)$ in time $2^{n/2+o(n)}$.
\end{corollary}
\begin{proof}
Combine the algorithm from Theorem~\ref{thm:smoothDGS} with the reduction from Theorem~\ref{thm:gapSVP} with $\eps = 2^{-n/4}$.
\end{proof}

\section{Other applications}
\label{sec:other}

\subsection{Approximating \problem{CVP} in \texorpdfstring{$2^{n+o(n)}$}{2 (n+o(n))} time}
\label{sec:approx-cvp}
\label{sec:approxCVP}

\begin{theorem}
\label{thm:cvptodgs}
For $\gamma = 1.97$, there is a reduction from $\gamma\text{-}\problem{CVP}$ to $\DGS{\frac{1}{2}}{}{2^{n/2}}$. The reduction makes $O(n^2)$ calls to the $\problem{DGS}$ oracle on an $(n+1)$-dimensional lattice and runs in time $2^{n/2} \cdot \poly(n)$.
\end{theorem}
\begin{proof}
On input $\lat = \lat(\vec{b}_1, \ldots, \vec{b}_n) \subset \R^n$ and $\vec{t} \in \R^n$, the reduction behaves as follows. It first uses Babai's nearest plane algorithm~\cite{Bab86} to approximate the distance to the lattice $\dist(\vec{t}, \lat)$, receiving as output $\approxd$. Fix $\delta := 1/n$. Then for $j = 1, \ldots, 10n^2$, let $s_j =  \approxd/(1+\delta)^j$, and let $\lat_j$ be the $(n+1)$-dimensional lattice generated by $(\vec{b}_i, 0)$ for $i \in [n]$, and the additional basis vector $(-\vec{t}, s_j)$. The reduction calls the ${\mbox{DGS}\xspace}$ oracle on $\lat_j$ with parameter $s_j$, and let $\vec{x}_j$ be the shortest vector among the returned vectors whose last coordinate is $s_j$. Finally, the reduction outputs the first $n$ coordinates of  $\vec{x}_j - (-\vec{t}, s_j)$ where $j$ is such that $\vec{x}_j$ is shortest.

The running time of the algorithm is clear. As was shown in~\cite{Bab86}, we have $d \leq \approxd \leq 2^{n/2} d$. Thus, there exists a $j$ such that 
$ \alpha d/\sqrt{n}  \leq s_j \leq (1+\delta) \alpha d/\sqrt{n}
$,
where $\alpha := \sqrt{2\pi/\log 2}$. 
Let $A = \{ \vec{y} \in \lat_j : \inner{\vec{y}, \vec{e}_{n+1}} = s_j\}$ be the set of vectors from which we choose $\vec{x}_j$. We note that it suffices to show that a sample from $D_{\lat_j, s_j}$ will land in $A$ and have length at most $\gamma d$ with probability at least $2^{-n/2 - O(1)}$. Indeed, if this is the case, then the algorithm will find a vector in $A$ of length at most $\gamma d$ with constant probability, and its output will be a $\gamma$-approximate closest vector.

We first consider the probability that a vector lands in $A$, $\rho_{s_j}(A)/\rho_{s_j}(\lat_j)$. For the denominator, using the fact that $\rho_s(\lat) \geq \rho_s(\lat + \vec{w})$ for any $\vec{w}$, we have
\begin{align*}
\rho_{s_j}(\lat_j)
&= \sum_{k = -\infty}^{\infty} \rho_{s_j}(k s_j) \rho_{s_j}(\lat + k \vec{t}) \\
&\le \rho_{s_j}(\lat) \sum_{k=-\infty}^{\infty} e^{-\pi k^2 }\\
&\leq 2 \rho_{s_j}(\lat)
\; .
\end{align*}
Turning to the numerator,
\begin{align*}
\rho_{s_j}( A ) &\ge   e^{-\pi(d^2 + s_j^2)/s_j^2} \cdot \rho_{s_j}(\lat) 
\ge   e^{-\pi(n/\alpha^2 + 1)} \cdot \rho_{s_j}(\lat)
\; .
\end{align*}
Thus, we have
\begin{equation}
\label{eq:inA}
\Pr_{\vec{X} \sim D_{\lat_j, s_j}}[\vec{X} \in A] 
\geq e^{-\pi n/\alpha^2}/100
\; .
\end{equation}

Set $t = \gamma/((1+\delta)\alpha)$. Recall from Lemma~\ref{lem:banaszczyktail} that
\begin{equation}
\label{eq:short}
\Pr_{\vec{X} \sim D_{\lat_j, s_j}}\big[\length{\vec{X}} > s_j t \sqrt{  n} \big] \leq \big( \sqrt{2 \pi e t^2} \exp(-\pi t^2) \big)^n
\; .
\end{equation}
Then, combining \eqref{eq:inA} and \eqref{eq:short}, plugging in the values for $\alpha$, $t$, and $\gamma$, and assuming $n$ is sufficiently large, gives
\begin{align*}
\Pr_{\vec{X} \sim D_{\lat_j, s_j}}\big[\vec{X} \in A,\ \length{\vec{X}} \leq s_j t \sqrt{n} \big] &\geq e^{-\pi n/\alpha^2}/100 - (2\pi e t^2)^{n/2} \cdot e^{-\pi t^2 n} \\
&\geq 2^{-n/2 - O(1)}
\; ,
\end{align*}
where we have used the fact that $e^{-\pi /\alpha^2} = 1/\sqrt{2}$ and $\sqrt{2\pi e t^2} \cdot e^{-\pi t^2} < 1/\sqrt{2}$.
The result follows from the fact that $s_j t \sqrt{n} \leq  (1+\delta) \alpha t d = \gamma d$.
\end{proof}

We note that the above proof actually yields a more general statement. In particular, for any $t > 1/\sqrt{2\pi}$, there is a reduction from $\gamma\text{-}\problem{CVP}$ to $\DGS{\eps}{}{M}$ where
\[
\gamma = \sqrt{\frac{2\pi t^2}{2\pi t^2 - \log(2\pi t^2) - 1}}\; ,
\]
and 
\[
M \approx \exp(\pi t^2 n/\gamma^2) = \exp(\pi n t^2)/(2\pi e t^2)^{n/2}
\; .
\]
We recover Theorem~\ref{thm:cvptodgs} by setting $t \approx 0.654$.

\begin{corollary}
\label{cor:approx-CVP}
There is a randomized algorithm that solves $1.97\text{-}\problem{CVP}$ in time $2^{n+o(n)}$.
\end{corollary}
\begin{proof}
Combine the reduction from Theorem~\ref{thm:cvptodgs} with the algorithm from Theorem~\ref{thm:DGS}.
\end{proof}

\subsection{Solving \texorpdfstring{$O(1)$-\problem{BDD}}{O(1)-BDD} in \texorpdfstring{$2^{n/2 + o(n)}$}{2 (n/2+o(n))} time}

Lyubashevsky and Micciancio show a polynomial-time reduction from $\frac{1}{2\gamma}\text{-}\problem{BDD}$ to $\gamma\text{-}\problem{GapSVP}$~\cite{LM09}. 
By combining this with Theorem~\ref{thm:gapSVP}, we immediately get a solution to $\alpha\text{-}\problem{BDD}$ for $\alpha \approx 1/4$. But, we can improve this to $\alpha \approx .422$ by using the following (slightly modified) theorem from \cite{cvpp} that shows how to solve a variant of \problem{BDD} directly using discrete Gaussian samples.

\begin{theorem}[{\cite[Theorem 3.1]{cvpp}}]
\label{thm:cvpp}
For any $\eps \in (0,1/200)$, let 
\[
\phi(\lat) :=  \frac{\sqrt{\log(1/\eps)/\pi - o(1)}}{2 \eta_\eps(\lat^*) }
\; .
\]
Then, there exists a reduction from
$\problem{CVP}^\phi$ to $\DGS{\frac{1}{2}}{\eta_{\eps}}{m}$ where $m = O(n \log (1/\eps)/\sqrt{\eps})$ and $\problem{CVP}^\phi$ is the problem of solving $\problem{CVP}$ for target vectors that are guaranteed to be within a distance $\phi(\lat)$ of the lattice. The reduction preserves the dimension, makes a single call to the $\problem{DGS}$ oracle, and runs in time $m \cdot \poly(n)$. 
\end{theorem}

\begin{corollary}\label{cor:bdd}
There is a randomized algorithm that solves $\alpha\text{-}\problem{BDD}$ in time $2^{n/2+o(n)}$ for $\alpha := .422 - o(1)$.
\end{corollary}
\begin{proof}
Let $\eps :=  2^{-n}$, and let $\phi(\lat) := \sqrt{\log(1/\eps)/\pi - o(1)}/(2\eta_\eps(\lat^*))$ as above. By Eq.~\eqref{eq:etalambda1small} of Lemma~\ref{lem:etalambda1}, any algorithm that solves $\problem{CVP}^\phi$ is also a solution to $\alpha\text{-}\problem{BDD}$ with
\[
\alpha := \frac{1}{2} \cdot \sqrt{\frac{\log(1/\eps)}{\log(1/\eps) + n \log \beta + o(n)}} > .422 - o(1)
\; .
\]
Applying Theorem~\ref{thm:cvpp} gives a reduction from $\alpha\text{-}\problem{BDD}$ to $\DGS{\frac{1}{2}}{\eta_\eps}{m}$ with $m = O(n \log (1/\eps)/\sqrt{\eps}) = 2^{n/2 + o(n)}$ that runs in time $m \cdot \poly(n)$. Finally, we note that Lemma~\ref{lem:etalambda1} implies that, $\sqrt{2} \eta_{1/2}(\lat) < \eta_{\eps}(\lat)$ for sufficiently large $n$. Therefore, Theorem~\ref{thm:smoothDGS} gives a solution to $\DGS{\exp(-\Omega(n))}{\eps}{m}$ with the desired running time.
\end{proof}

\subsection{Approximating SIVP in \texorpdfstring{$2^{n/2 + o(n)}$}{2 (n/2+o(n))} time}

We use the following lemma, which is a slight variant of~\cite[Lemma 3.17]{oded05}
combined with Lemma 2.12 there.

\begin{lemma}
\label{lem:sivp}
There is a polynomial-time reduction from $\gamma\text{-}\problem{SIVP}$ to $\DGS{\frac{1}{2}}{2\eta_\eps}{}$ where $\gamma := O(\sqrt{n \log n})$ and $\eps := 1/10$.
\end{lemma}

\begin{corollary}
\label{cor:sivp}
There is a randomized algorithm that solves $\gamma\text{-}\problem{SIVP}$ in time $2^{n/2+o(n)}$ where $\gamma := O(\sqrt{n \log n})$.
\end{corollary}
\begin{proof}
Combine the reduction from Lemma~\ref{lem:sivp} with the algorithm from Theorem~\ref{thm:smoothDGS}.
\end{proof}

\newcommand{\etalchar}[1]{$^{#1}$}

\end{document}